%% file: main.tex
\documentclass{article}

\usepackage{microtype}
\usepackage{graphicx}
\usepackage{subcaption}
\usepackage{booktabs} 

\usepackage{hyperref}

\usepackage[accepted]{icml2024}

\usepackage{amsmath}
\usepackage{amssymb}
\usepackage{mathtools}
\usepackage{amsthm}

\usepackage[capitalize,noabbrev]{cleveref}

\theoremstyle{plain}
\newtheorem{theorem}{Theorem}[section]

\theoremstyle{definition}
\newtheorem{definition}[theorem]{Definition}

\theoremstyle{remark}

\input{0_custom_declares}

\icmltitlerunning{An Algorithm for Streaming Differentially Private Data}

\begin{document}

\twocolumn[
\icmltitle{An Algorithm for Streaming Differentially Private Data}



\icmlsetsymbol{equal}{*}

\begin{icmlauthorlist}
\icmlauthor{Girish Kumar}{ucd}
\icmlauthor{Thomas Strohmer}{ucd}
\icmlauthor{Roman Vershynin}{uci}
\end{icmlauthorlist}

\icmlaffiliation{ucd}{Department of Mathematics, University of California, Davis, USA}
\icmlaffiliation{uci}{Department of Mathematics, University of California, Irvine, USA}

\icmlcorrespondingauthor{Girish Kumar}{gkum@ucdavis.edu}

\icmlkeywords{privacy, stream}

\vskip 0.3in
]



\printAffiliationsAndNotice{}  

\begin{abstract}
    Much of the research in differential privacy has focused on offline applications with the assumption that all data is available at once. When these algorithms are applied in practice to streams where data is collected over time,  this either violates the privacy guarantees or results in poor utility. We derive an algorithm for differentially private synthetic streaming data generation, especially curated towards spatial datasets. Furthermore, we provide a general framework for online selective counting among a collection of queries which forms a basis for many tasks such as query answering and synthetic data generation. The utility of our algorithm is verified on both real-world and simulated datasets. 
\end{abstract}

\input{0_intro_related_work}
\input{1_setup}

\input{2_3_method}
\input{4_counters}
\input{5_experiments}

\section*{Acknowledgements}
G.K.\ and T.S.\ acknowledge support from  NSF DMS-2027248, NSF DMS-2208356, and NIH R01HL16351.  R.V.\ acknowledges support from NSF DMS-1954233, NSF DMS-2027299, U.S. Army 76649-CS, and NSF+Simons Research Collaborations on the Mathematical and Scientific Foundations of Deep Learning.



\bibliography{bibfile}
\bibliographystyle{icml2024}

\newpage
\input{6_appendix}

\end{document}

%% file: 0_custom_declares.tex
\usepackage{amsfonts, dsfont, xcolor, commath}

\DeclareMathOperator{\Lap}{Lap}

\DeclareMathOperator{\children}{children}

\DeclareMathOperator{\Ext}{Ext}
\DeclareMathOperator{\PrivTree}{PrivTree}
\DeclareMathOperator{\Counting}{Counting}

\newcommand{\inparanth}[1]{\left( #1 \right)}
\newcommand{\prob}[1]{\mathbb{P} \left\{ #1 \right\}}
\newcommand{\bigo}[1]{\mathcal{O}\left( #1 \right)}

\def \N {\mathbb{N}}
\def \P {\mathbb{P}}
\def \R {\mathbb{R}}

\def \AA {\mathcal{A}}
\def \CC {\mathcal{C}}

\def \LL {\mathcal{L}}
\def \MM {\mathcal{M}}

\def \PP {\mathcal{P}}

\def \a {\alpha}
\def \b {\beta}

\def \e {\varepsilon}
\def \d {\delta}

\def \one {{\textbf 1}}
\def \ind {{\mathds 1}}

\newcommand\restr[2]{{\left.\kern-\nulldelimiterspace{#1}\right|_{#2}}}
  

%% file: 0_intro_related_work.tex
\section{Introduction}


Many data driven applications require frequent access to the user's location to offer their services more efficiently. These are often called Location Based Services (LBS). Examples of LBS include queries for nearby businesses \cite{lbs_advertisement}, calling for taxi pickup \cite{lbs_local_taxi, lbs_taxi_clustering}, and local weather information \cite{lbs_weather_app_wiki}.

Much of location data contains sensitive information about a user in itself or when cross-referenced with additional information. Several studies \cite{attack_survey} have shown that publishing location data is susceptible to revealing sensitive details about the user and simply de-identification \cite{attack_ny_taxi} or aggregation \cite{attack_aggregdated_data_ash} is not sufficient. Thus, privacy concerns limit the usage and distribution of sensitive user data. One potential solution to this problem is privacy-preserving synthetic data generation. 

Differential privacy (DP) can be used to quantifiably guarantee the privacy of an individual and it has been used by many notable institutions such as US Census \cite{Abowd2019CensusTD}, Google \cite{rapporGoogle}, and Apple \cite{appleDP}. DP has seen a lot of research across multiple applications such as statistical query answering \cite{mckenna2021hdmm}, regression \cite{dpLinearRegression}, clustering \cite{dp_dbscan}, and large-scale deep learning \cite{abadi2016deep}. Privacy-preserving synthetic data generation has also been explored for various data domains such as tabular microdata \cite{hardt2012mwem, zhang2017privbayes, pate_gan, tao2021benchmarking, mckenna2021winning}, natural language \cite{li_xuechen_llm_dp, yue2022synthetic}, and images \cite{dp_cgan, xu2019ganobfuscator}.

DP has also been explored extensively for various use cases concerning spatial datasets such as answering range queries \cite{privtree}, collecting user location data \cite{location_dp_survey}, and collecting user trajectories \cite{li2017trajectoryDP}. In particular for synthetic spatial microdata generation, a popular approach is to learn the density of true data as a histogram over a Private Spatial Decomposition (PSD) of the data domain and sample from this histogram \cite{privtree, psd_kim2018differentially, psd_maryam2018, psd_qardaji2013ug}. Our work also uses this approach and builds upon the method PrivTree \cite{privtree} which is a very effective algorithm for generating PSD.

Despite a vast body of research, the majority of developments in differential privacy have been restricted to a one-time collection or release of information. In many practical applications, the techniques are required to be applied either on an event basis or a regular interval. Many industry applications of DP simply re-run the algorithm on all data collected so far, thus making the naive (and usually, incorrect) assumption that future data contributions by users are completely independent of the past. This either violates the privacy guarantee completely or results in a very superficial guarantee of privacy \cite{tang2017ApplePrivacy}.

Differential privacy can also be used for streaming data allowing the privacy-preserving release of information in an online manner with a guarantee that spans over the entire time horizon \cite{Dwork2010ContinualDP}. However, most existing algorithms using this concept such as \cite{Dwork2010ContinualDP, Chan2010ContinualPrivateStats, Joseph2018LocalDP, Ding2017CollectingTD} are limited to the release of a statistic after observing a stream of one-dimensional input (typically a bit stream) and have not been explored for tasks such as synthetic data generation.

In this work, we introduce the novel task of privacy-preserving synthetic multi-dimensional stream generation. We are particularly interested in low-sensitivity granular spatial datasets collected over time. For motivation, consider the publication of the coordinates (latitude and longitude) of residential locations of people with active coronavirus infection. As people get infected or recover, the dataset evolves and the density of infection across our domain can dynamically change over time. Such a dataset can be extremely helpful in making public policy and health decisions by tracking the spread of a pandemic over both time and space. This dataset has low sensitivity in the sense that it is rare for a person to get infected more than a few times in, say, a year.
We present a method that can be used to generate synthetic data streams with differential privacy and we demonstrate its utility on real-world datasets. Our contributions are summarized below.
\vspace*{-2mm}
\begin{enumerate}
\setlength{\itemsep}{-0.3ex}
    \item To the best of our knowledge, we present the first differentially private streaming algorithm for the release of multi-dimensional synthetic data;
    \item we present a meta-framework that can be applied to a large number of counting algorithms to resume differentially private counting on regular intervals;
    \item we further demonstrate the utility of this algorithm for synthetic data generation on both simulated and real-world spatial datasets;
    \item furthermore, our algorithm can handle both the addition and the deletion of data points (sometimes referred to as turnstile model), and thus dynamic changes of a dataset over time.
\end{enumerate}

\subsection{Related Work}
A large body of work has explored privacy-preserving techniques for spatial datasets. The methods can be broadly classified into two categories based on whether or not the curator is a trusted entity. If the curator is not-trusted, more strict privacy guarantees such as Local Differential Privacy and Geo-Indistinguishability are used and action is taken at the user level before the data reaches the server. We focus on the case where data is stored in a trusted server and the curator has access to the true data of all users. This setup is more suited for publishing privacy-preserving microdata and aggregate statistics. Most methods in this domain rely on Private Spatial Decomposition (PSD) to create a histogram-type density estimation of the spatial data domain. PrivTree \cite{privtree} is perhaps the most popular algorithm of this type and we refer the reader to \cite{location_dp_survey} for a survey of some other related methods. However, these offline algorithms assume we have access to the entire dataset at once so they do not apply directly to our use case.

We use a notion of differential privacy that accepts stream as input and guarantees privacy over the entire time horizon, sometimes referred to as continual DP. Early work such as \cite{Dwork2010ContinualDP} and \cite{Chan2010ContinualPrivateStats} explores the release of bit count over an infinite stream of bits and introduce various effective algorithms, including the Binary Tree mechanism. We refer to these mechanisms here as \textit{Counters} and use them as a subroutine of our algorithm. In \cite{ContinualRectQueries}, the authors build upon the work in \cite{Dwork2010ContinualDP} and use the Binary Tree mechanism together with an online partitioning algorithm to answer range queries. However, the problem considered is answering queries on offline datasets. In \cite{cdp_wang2021continuousStream, cdp_chen2017pegasus} authors further build upon the task of privately releasing a stream of bits or integers under user and event-level privacy. A recent work \cite{cdp_jain2023countingDistinct} addresses deletion when observing a stream under differential privacy. They approach the problem of releasing a count of distinct elements in a stream. Since these works approach the task of counting a stream of one-dimensional data, they are different from our use case of multi-dimensional density estimation and synthetic data generation. The Binary Tree Mechanism has also been used in other problems such as online convex learning \cite{continual_private_bandit} and deep learning \cite{private_dl_without_sampling}, under the name \textit{tree-based aggregation trick}. Many works with streaming input have also explored the use of local differential privacy for the collection and release of time series data such as telemetry \cite{Joseph2018LocalDP, Ding2017CollectingTD}. 

To the best of our knowledge a very recent work \cite{cdp_bun2023continualSynthData} is the only other to approach the task of releasing a synthetic stream with differential privacy. However, they approach a very different problem where the universe consists of a fixed set of users, each contributing to the dataset at all times. Moreover, a user's contribution is limited to one bit at a time, and the generated synthetic data is derived to answer a fixed set of queries. In contrast, we allow multi-dimensional continuous value input from an arbitrary number of users and demonstrate the utility over randomly generated range queries.

%% file: 1_setup.tex
\section{The problem}

In this paper, we present an algorithm that transforms a data stream into a differentially private data stream. The algorithm can handle quite general data streams: at each time $t \in \N$, the data is a subset of some abstract set $\Omega$. For instance, if $\Omega$ can be the location of all U.S. hospitals, and the data at time $t=3$ can be the locations of all patients spending time in hospitals on day $3$. Such data can be conveniently represented by a {\em data stream}, which is any function of the form $f(x,t): \Omega \times \N \to \R$. We can interpret $f(x,t)$ as the number of data points present at location $x \in \Omega$ at time $t \in \N$. For instance, $f(x,3)$ can be the number of COVID positive patients at location $x$ on day $3$. 

We present {\em an $\e$-differentially private, streaming algorithm that takes as an input a data stream and returns as an output a data stream.} This algorithm tries to make the output stream as close as possible to the input data stream, while upholding differential privacy. 

\subsection{A DP streaming, synthetic data algorithm}	\label{s: dp streaming synthetic}

The classical definition of differential privacy of a randomized differential algorithm $\AA$ 
demands that for any pair of input data $f$, $\tilde{f}$ that differ by a single data point, 
the outputs $\AA(f)$ and $\AA(\tilde{f})$ be statistically indistinguishable. 
Here is a version of this definition, which is very naturally adaptable to problems about data streams.

\begin{definition}[Differential privacy]		\label{def: DP}
  A randomized algorithm $\AA$ that takes as an input a point $f$ in some given normed space is $\e$-differentially private if for any two data streams that satisfy $\norm[0]{f-\tilde{f}}=1$, the inequality
  \begin{equation}	\label{eq: DP}
  \P\{\AA(\tilde{f}) \in S\} \le e^\e \cdot \P\{\AA(f) \in S\}
  \end{equation}
  holds for any measurable set of outputs $S$.
\end{definition}

To keep track of the change of data stream $f$ over time, 
it is natural to consider the {\em differential stream} 
\begin{equation}	\label{eq: differential stream}
	\nabla f(x,t) = f(x,t)-f(x,t-1), \quad t \in \N,
\end{equation}
where we set $f(x,0)=0$. The total change of $f$ over all times and locations is the quantity
$$
\norm{f}_\nabla \coloneqq \sum_{x \in \Omega} \sum_{t \in \N} \abs{\nabla f(x,t)},
$$ 
which defines a seminorm on the space of data streams.
A moment's thought reveals that two data streams satisfy 
\begin{equation}	\label{eq: sensitivity}
\norm[0]{f-\tilde{f}}_\nabla = 1
\end{equation}
if and only if $\tilde{f}$ can be obtained from $f$ by changing {\em a single data point}: either one data point is added at some time and is never removed later, or one data point is removed and is never added back later. This makes it natural to consider DP of streaming algorithms wrt.\ the $\norm{\cdot}_\nabla$ norm.

The algorithm we are about to describe produces {\em synthetic data}: it converts an input stream $f(x,t)$ into an output stream $g(x,t)$. The algorithm is {\em streaming}: at each time $t_0$, it can only see the part of the input stream $f(x,t)$ for all $x \in \Omega$ and $t \le t_0$, and at that time the algorithm outputs $g(x,t_0)$ for all $x \in \Omega$.

\subsection{Privacy of multiple data points is protected}
Now that we quantified the effect of the change of a single input data point, we can change any number of input data points. If two data streams satisfy $\norm[0]{f-\tilde{f}}_\nabla = k$, then applying Definition~\ref{def: DP} $k$ times and using the triangle inequality, we conclude that 
$$
\P\{\AA(\tilde{f}) \in S\} \le e^{k\e} \cdot \P\{\AA(f) \in S\}.
$$
For example, suppose that a patient who gets sick and spends a week at some hospital; then she recovers, but after some time she gets sick again and spends another week at another hospital and finally recovers completely. If the data stream $\tilde{f}$ is obtained from $f$ by removing such a patient, then, due to the four events described above,  $\norm[0]{f-\tilde{f}}_\nabla = 4$.
Hence, we conclude that
$
\P\{\AA(\tilde{f}) \in S\} \le e^{4\e} \cdot \P\{\AA(f) \in S\}.
$
In other words, the privacy of patients who contribute four events to the data stream is automatically protected as well, although the protection guarantee is four times weaker than for patients who contribute a single event. 

%% file: 2_3_method.tex
\section{The method}

Here, we describe our method in broad brushstrokes. In Appendix~\ref{app:optim} we discuss how we optimize the computational and storage cost of our algorithm.

\subsection{From streaming on sets to streaming on trees}

First, we convert the problem of differentially private streaming of a function $f(x,t)$ on a set $\Omega$ to a problem of differentially private streaming of a function $F(x,t)$ on a {\em tree}. To this end, fix some {\em hierarchical partition} of the domain $\Omega$. Thus, assume that  $\Omega$ is partitioned into some $\beta>1$ subsets $\Omega_1,\ldots,\Omega_\beta$, and each of these subsets $\Omega_i$ is partitioned into $\beta$ further subsets, and so on. A hierarchical partition can be equivalently represented by a {\em tree} $T$ whose vertices are subsets of $\Omega$ and the children of each vertex form a partition of that vertex. Thus, the tree $T$ has root $\Omega$; the root is connected to the $\beta$ vertices $\Omega_1,\ldots,\Omega_\beta$, and so on. We refer to $\beta$ as the {\em fanout} number of the tree.

In practice, there often exists a natural hierarchical decomposition of $\Omega$. For example, if $\Omega=\{0,1\}^d$ a binary partition obtained by fixing a coordinate is natural such as $\Omega_1=\{0\} \times \{0,1\}^{d-1}$ and $\Omega_2=\{1\} \times \{0,1\}^{d-1}$. Each of $\Omega_1$ and $\Omega_2$ can be further partitioned by fixing another coordinate. As another example, if $\Omega=[0,1]^d$ and fanout $\beta=2$, a similar natural partition can be obtained by splitting a particular dimension's region into halves such that $\Omega_1=[0,\frac{1}{2})\times[0,1]^{d-1}$ and $\Omega_2=[\frac{1}{2},1]\times[0,1]^{d-1}$. Each of $\Omega_1$ and $\Omega_2$ can be further partitioned by splitting another coordinate's region into halves.

We can convert any function on the set $\Omega$ into a function on the vertices of the tree $T$ by summing the values in each vertex. I.e., to convert $f \in \R^\Omega$ into $F \in \R^{V(T)}$, we set
\begin{equation}	\label{eq: contraction}
F(v) \coloneqq \sum_{x \in v} f(x), 
\quad v \in V(T).
\end{equation}
Vice versa, we can convert any function $G$ on the vertices of the tree $T$ into a function $g$ on the set $\Omega$ by assigning value $G(v)$ to one arbitrarily chosen point in each leaf $v$. In practice, however, the following variant of this rule works better if $G(v)>0$. 
Assign value $1$ to $\lceil G(v) \rceil$ random points in $v$, i.e. set
\begin{equation}	\label{eq: contraction reversed}
g \coloneqq \sum_{v \in \LL(T)} \sum_{i=1}^{\lceil G(v) \rceil} \one_{x_i(v)}
\end{equation}
where $x_i(v)$ are independent random points in $v$ and $\one_x$ denotes the indicator function of the set $\{x\}$. The points $x_i(v)$ can be sampled from any probability measure on $v$, and in practice we often choose the uniform measure on $v$.

Summarizing, we reduced our original problem to constructing an algorithm that transforms any given stream 
$F(x,t): V(T) \times \N \to \{0,1,2,\ldots\}$
into a differentially private stream $G(x,t): V(T) \times \N \to \{0,1,2,\ldots\}$
where $V(T)$ is the vertex set of a fixed, known tree $T$.

\subsection{Consistent extension}

Let $C(T)$ denote the set of all functions $F \in \R^{V(T)}$ that can be obtained from functions $f \in \R^\Omega$ using transformation \eqref{eq: contraction}. The transformation is linear, so $C(T)$ must be a linear subspace of $\R^{V(T)}$. A moment's thought reveals that $C(T)$ is comprised of all {\em consistent functions} -- the functions $F$ that satisfy the equations\
\begin{equation}	\label{eq: consistency}
F(v) = \sum_{u \in \children(v)} F(u)
\quad \text{for all } v \in V(T).
\end{equation}
Any function on $V(T)$ can be transformed into a consistent function by pushing the values up the tree and spreading them uniformly down the tree. More specifically, this can be achieved by the linear transformation 
$$
\Ext_T : \R^{V(T)} \to C(T),
$$
that we call the {\em consistent extension}. Suppose that a function $F$ takes value $1$ on some vertex $v \in V(T)$ and value $0$ on all other vertices. To define $G=\Ext_T(F)$, we let $G(u)$ equal $1$ for any ancestor of $v$ including $v$ itself, $1/\beta$ for any child of $v$, $1/\beta^2$ for any grandchild of $v$, and so on. In other words, we set $G(u)=\beta^{-\max(0,d(v,u))}$ where $d(v,u)$ denotes the directed distance on the tree $T$, which equals the usual graph distance (the number of edges in the path from $v$ to $u$) if $u$ is a descendant of $v$, and minus the graph distance otherwise. Extending this rule by linearity, we arrive at the explicit definition of the consistent extension operator:
$$
\Ext_T(F)(u) \coloneqq \sum_{v \in V(T)} F(v) \, \beta^{-\max(0,d(v,u))}, 
\quad u \in V(T).
$$

%
%

By definition~\eqref{eq: consistency}, a consistent function is uniquely determined by its values on the leaves of the tree $T$. Thus a natural norm of $C(T)$ is 
$$
\norm{F}_{C(T)} \coloneqq \sum_{v \in \LL(T)} \abs{F(v)}.
$$

\subsection{Differentially Private Tree}		\label{s: PrivTree subsection}

A key subroutine of our method is a version of the remarkable algorithm $\PrivTree$ due to \cite{privtree}. In the absence of noise addition, one can think of $\PrivTree$ as a deterministic algorithm that inputs a tree $T$ and a function $F \in \R^{V(T)}$ and outputs a subtree of $T$. The algorithm grows the subtree iteratively: for every vertex $v$, if $F(v)$ is larger than a certain threshold $\theta$, the children of vertex $v$ are added to the subtree. 

\begin{theorem}[Privacy of PrivTree]\label{thm: PrivTree privacy}
  The randomized algorithm $\MM \coloneqq \PrivTree_T(F, \e, \theta)$ is $\e$-differentially private in the $\norm{\cdot}_{C(T)}$ norm for any $\theta \geq 0$.
\end{theorem}

In Appendix~\ref{app: PrivTree proof}, we give the $\PrivTree$ algorithm and proof of its delicate privacy guarantees in detail.

\subsection{Differentially Private Stream}

We present our method PHDStream (Private Hierarchical Decomposition of Stream) in Algorithm~\ref{alg: PHDStream}. Algorithm~\ref{alg: PHDStreamTree} transforms an input stream $F(\cdot,t) \in V(T) \times \N \to \R$ into a stream $G(\cdot,t) \in V(T) \times \N \to \R$. In the algorithm, $\LL(T)$ denotes the set of leaves of tree $T$, and $\Lap(\cdot, t, 2/\e)$ denote independent Laplacian random variables. Using Algorithm~\ref{alg: PHDStreamTree} as a subroutine, Algorithm~\ref{alg: PHDStream} transforms a stream $f(x,t) \in \Omega \times \N \to \R$ into a stream 
$g(x,t) \in \Omega \times \N \to \R$:

\begin{algorithm}[tb]
    \begin{algorithmic}[1]
        \STATE {\bfseries Input:} $F \in C(T)$, the privacy budget parameter $\e$, the threshold count for a node $\theta$.
        \STATE {\bfseries Output:} A stream $G \in C(T)$.
        \STATE Initialize $G(v,0)=0$ for all $v \in V(T)$.
        \FOR{every time $t \in \N$}
            \STATE $T(t) \leftarrow \PrivTree_T \left( G(\cdot,t-1) + \nabla F(\cdot,t), \e/2, \theta \right)$.
            \STATE $d(v,t) \leftarrow 
        	\begin{cases}
        		\nabla F(v,t) + \Lap(\cdot, t, 2/\e) & \text{if } v \in \LL(T(t)) \\
        		0 & \text{otherwise}.
        	\end{cases}$ 
            \STATE $G(v,t) \leftarrow G(v,t-1) + \Ext_T(d(\cdot,t))	\quad \forall v \in V(T)$.
        \ENDFOR
    \end{algorithmic}
    \caption{PHDStreamTree}
    \label{alg: PHDStreamTree}
\end{algorithm}

\begin{algorithm}[ht]
    \begin{algorithmic}[1]
        \STATE {\bfseries Input:} Input data stream $f:\Omega \times \N \to \R$, the privacy budget parameter $\e$, the threshold count for a node $\theta$.
        \STATE {\bfseries Output:} Data stream $g$.
        \STATE Apply PHDStreamTree (Algorithm~\ref{alg: PHDStreamTree}) for $F$ obtained from $f$ using \eqref{eq: contraction}, and we convert the output $G$ into $g$ using \eqref{eq: contraction reversed}.
    \end{algorithmic}
    \caption{PHDStream}
    \label{alg: PHDStream}
\end{algorithm}

\subsection{Privacy of PHDStream}

In the first step of this algorithm, we consider the previously released synthetic stream $G(\cdot,t-1)$, update it with the newly received real data $\nabla F(\cdot,t)$, and feed it into $\PrivTree_T$, which produces a differentially private subtree $T(t)$ of $T$.

In the next two steps, we compute the updated stream on the tree. It is tempting to choose for $G(\cdot,t)$ a stream computed by a simple random perturbation as, $G(\cdot,t) = G(\cdot,t-1) + \nabla F(\cdot,t) + \Lap(\cdot, t, 2/\e)$. The problem, however, is that such randomly perturbed stream would not be differentially private. Indeed, imagine we make a stream $\tilde{f}$ by changing the value of the input stream $f$ at some time $t\in \N$ and point $x \in \Omega$ by $1$. Then the sensitivity condition \eqref{eq: sensitivity} holds. But when we convert $f$ into a consistent function $F$ on the tree, using \eqref{eq: contraction}, that little change propagates up the tree. It affects the values of $F$ not just at one leaf $v \ni x$ but all of the ancestors of $v$ as well, and this could be too many changes to protect. In other words, consistency on the tree makes sensitivity too high, which in turn jeopardizes privacy.

To halt the propagation of small changes up the tree, the last two steps of the algorithm restrict the function {\em only} on the leaves of the subtree. This restriction controls sensitivity: a change to $F(v)$ made in one leaf $v$ does not propagate anymore, and the resulting function $d(v,t)$ is differentially private. In the last step, we extend the function $d(v,t)$ from the leaves to the whole tree--an operation that preserves privacy--and use it as an update to the previous, already differentially private, synthetic stream $G(\cdot,t-1)$. 

These considerations lead to the following privacy guarantee as announced in Section~\ref{s: dp streaming synthetic}, i.e. in the sense of Definition~\ref{def: DP}, for the $\norm{\cdot}_\nabla$ norm.

\begin{theorem}[Privacy of PHDStream]		\label{thm: PHDStream privacy}
	The PHDStream algorithm is $\e$-differentially private.
\end{theorem}

A formal proof of Theorem~\ref{thm: PHDStream privacy} is given in Appendix~\ref{app: PHDStream proof}.

%% file: 4_counters.tex
\section{Counters and selective counting}
A key step of Algorithm~\ref{alg: PHDStreamTree} at any time $t$ is Step~{6} where we add noise to the leaves of the subtree $T(t)$. Consider a node $v \in V(T)$ that becomes a leaf in the subtree $T(t)$ for times $t \in N_v \subseteq N$. In the algorithm, we add an independent noise at each time in $N_v$. Focusing only on a particular node, can we make this counting more efficient? The question becomes: given an input stream of values for a node $v$ as $\nabla F(v, t)$ for $t \in N_v$, can we find an output stream of values $d(v, \cdot)$ in a differentially private manner while ensuring that the input and output streams are close to each other?

The problem of releasing a stream of values with differential privacy has been well studied in the literature \cite{Dwork2010ContinualDP, Chan2010ContinualPrivateStats}. We will use some of these known algorithms at each node $v$ of the tree to perform the counting more efficiently. Let us first introduce the problem more generally using the concept of {\em Counters}.

\subsection{Counters}
\begin{definition}\label{def: Counter}
    An $(\a, \d)$-accurate counter $\CC$ is a randomized streaming algorithm that estimates the sum of an input stream of values $f:\N \to\R$ and maps it to an output stream of values $g: \N\to\R$ such that for each time $t\in\N$, 
    $$\mathbb{P} \biggl\{ \bigg\lvert g(t)- \sum_{t'\leq t} f(t') \bigg\rvert \leq \a(t, \d) \biggr\} \geq 1-\d,$$
    where the probability is over the randomness of $\CC$ and $\d$ is a small constant.
\end{definition}

Furthermore, we are interested in counters that satisfy differential privacy guarantees as per Definition~\ref{def: DP} with respect to the norm $\norm{f} \coloneqq \sum_{t \in \N} \abs{f(t)}$. Note that we have not used the norm $\norm{\cdot}_\nabla$ here as the input to the counter algorithm will already be the differential stream $\nabla f$. The foundational work of \cite{Chan2010ContinualPrivateStats} and \cite{Dwork2010ContinualDP} introduced private counters for the sum of bit-streams but the same algorithms can be applied to real-valued streams as well. In particular, we will be using the Simple II, Two-Level, and Binary Tree algorithms from \cite{Chan2010ContinualPrivateStats}, hereafter referred to as Simple, Block, and Binary Tree Counters respectively. The counter algorithms are included in Appendix \ref{app:counters} for completeness. We also discuss there how the Binary Tree counter is only useful if the input stream to the counter has a large time horizon and thus we limit our experiments in Section~\ref{s: experiments and results} to Simple and Block counters.

\subsection{PHDStreamTree with counters}
Algorithm~\ref{alg: PHDStreamTree with counters} is a version of Algorithm~\ref{alg: PHDStreamTree} with counters. Here, we create an instance of some counter algorithm for each node $v \in V(T)$. Algorithm~\ref{alg: PHDStreamTree with counters} is agnostic to the counter algorithm used, and we can even use different counter algorithms for different nodes of the tree. Since at any time $t \in \N$, we only count at the leaf nodes of the tree $T(t)$, the counter $C_v$ is updated for any node $v\in V(T)$ if and only if $v \in \LL(T(t))$. Hence the input to the counter $C_v$ is the restriction of the stream $\nabla F$ on the set of times in $N_v$, that is $\restr{\nabla F(v, \cdot)}{N_v}$, where $N_v=\{t\in\N \mid v\in \LL(T(t))\}$. In the subsequent sections, we discuss a few general ideas about the use of multiple counters and selectively updating some of them. We use these discussions to prove the privacy of Algorithm~\ref{alg: PHDStreamTree with counters} in Subsection~\ref{s: privacy of PHDStreamTree with counters}.

\begin{algorithm}[ht]
    \begin{algorithmic}[1]
        \STATE {\bfseries Input:} $F \in C(T)$, privacy budget parameter $\e$, threshold count for a node $\theta$.
        \STATE {\bfseries Output:}  A stream $G \in C(T)$.
        \STATE Initialize $G(v,0)=0$ for all $v \in V(T)$.
        \STATE Initialize a counter $\CC_v$ with privacy budget $\e/2$ for all nodes $v \in V(T)$.
        \FOR{every time $t \in \N$}
            \STATE $T(t) \leftarrow \PrivTree_T \left( G(\cdot,t-1) + \nabla F(\cdot,t), \e/2, \theta \right)$.
            \STATE $d(v,t) \leftarrow 
        	\begin{cases}
        		\CC_v(\nabla F(v,t)) & \text{if } v \in \LL(T(t)) \\
        		0 & \text{otherwise}.
        	\end{cases}$
            \STATE $G(v,t) \leftarrow G(v,t-1) + \Ext_T(d(\cdot,t))	\quad \forall v \in V(T)$.
        \ENDFOR
    \end{algorithmic}
    \caption{PHDStreamTree with counters}
    \label{alg: PHDStreamTree with counters}
\end{algorithm}

\subsection{Multi-dimensional counter}
To efficiently prove the privacy guarantees of our algorithm, which utilizes many counters, we introduce the notion of a multi-dimensional counter. A $d$-dimensional counter $\CC$ is a randomized streaming algorithm consisting of $d$ independent counters $\CC_1, \CC_2, \ldots, \CC_d$ such that it maps an input stream $f:\N\to\R^d$ to an output stream $g:\N\to\R^d$ with $g(t) = \inparanth{\CC_1(f(t)_1), \CC_2(f(t)_2), \ldots, \CC_d(f(t)_d)}$ for all $t \in \N$. For differential privacy of such counters, we will still use Definition~\ref{def: DP} but with the extension of the norm to streams with multi-dimensional output as $\norm{f} \coloneqq \sum_{t \in \N} \norm{f(t)}_{\ell_1}$.

\subsection{Selective Counting}
Let us assume that we have a set of $k$ counters $\set{\CC_1, \CC_2, \dots, \CC_k}$. At any time $t\in\N$, we want to activate exactly one of these counters as selected by a randomized streaming algorithm $\MM$. The algorithm $\MM$ depends on the input stream $f$ and optionally on the entire previous output history, that is the time indices when each of the counters was selected and their corresponding outputs at those times. We present this idea formally as Algorithm~\ref{alg:selective_counting}.

\begin{algorithm}[h]
    \begin{algorithmic}[1]
        \STATE {\bfseries Input:} An input stream $f$, a set  $\{\CC_1, \CC_2, \dots, \CC_k\}$ of $k$ counters, a counter-selecting differentially private algorithm $\MM$.
        \STATE {\bfseries Output:}  An output stream $g$.
        \STATE Initialize $J_l \leftarrow \emptyset, \forall\, l \in \{1, 2, \dots, k\}$.
        \FOR{$t=1$ to $\infty$}
            \STATE $l_t \leftarrow \MM (f(t), g(t-1))$.
            \STATE $J_{l_t} \leftarrow J_{l_t} \cup \set{t}$; Add time index $t$ for counter $l_t$.
            \STATE Set $g(t) \leftarrow \inparanth{l_t, \CC_{l_t} (f(t))}$; Selected counter and updated count.
        \ENDFOR
    \end{algorithmic}
    \caption{Online Selective Counting}
    \label{alg:selective_counting}
\end{algorithm}

\begin{theorem}{(Selective Counting)} \label{thm:selective_counting}
Let $\CC_i, i\in \{1, 2, \dots, k\}$ be $\e_C$-DP. Let $\MM$ be a counter-selecting streaming algorithm that is $\e_M$-DP. Then, Algorithm~\ref{alg:selective_counting} is $(\e_M+\e_C)$-DP.
\end{theorem}

The proof of Theorem~\ref{thm:selective_counting} is given in Appendix~\ref{app: proof selective counting}. Note the following about Algorithm~\ref{alg:selective_counting} and Theorem~\ref{thm:selective_counting}: (1)~each individual counter can be of any finite dimensionality, (2)~we do not assume anything about the relation among the counters, (3)~the privacy guarantees are independent of the number of counters, and (4)~the algorithm can be easily extended to the case when the input streams to all sub-routines $\CC_1, \CC_2, \ldots, \CC_k$, and $\MM$ may be different.

\subsection{Privacy of PHDStreamTree with counters}
\label{s: privacy of PHDStreamTree with counters}
We first show that Algorithm~\ref{alg: PHDStreamTree with counters} is a version of the selective counting algorithm (Algorithm~\ref{alg:selective_counting}). Let $\PP(T)$ be a set of all possible subtrees of the tree $T$. Since each node $v \in T$ is associated with a counter, the collection of counters in $\LL(S)$ for any $S \in \PP(T)$ is a multi-dimensional counter. Moreover, for any neighboring input streams $F$ and $\tilde F$ with $\norm{F-\tilde F}_{\nabla}=1$, at most one leaf node will differ in the count. Hence $\LL(S)$ is a counter with $\e/2$-DP. $\PP(T)$ is thus a set of multi-dimensional counters, each satisfying $\e/2$-DP. $\PrivTree_T$ at time $t$ selects one counter as $\LL(T(t))$ from $\PP(T)$ and performs counting. Hence, by Theorem~\ref{thm:selective_counting}, Algorithm~\ref{alg: PHDStreamTree with counters} is $\e$-DP.

%% file: 5_experiments.tex
\section{Experiments and results}\label{s: experiments and results}

 \begin{figure*}[h!]
    \vskip 0.1in
     \centering
     \begin{subfigure}{\textwidth}
         \centering
         \includegraphics[width=\textwidth]{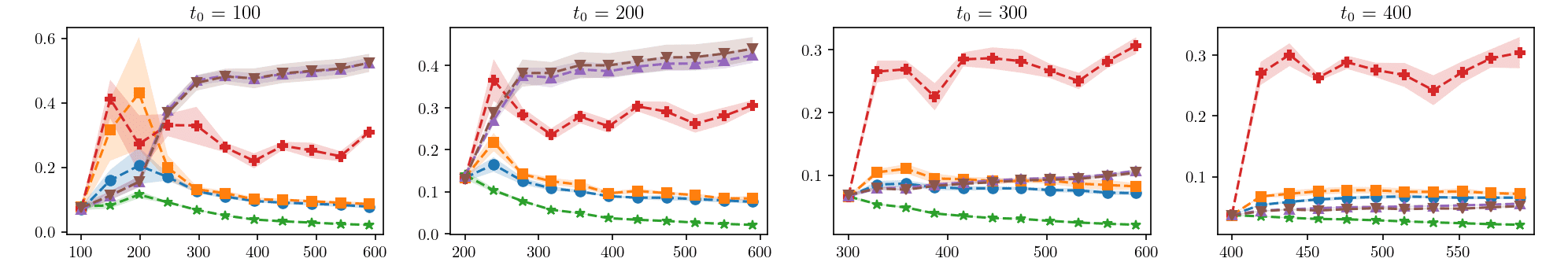}
        \caption{Gowalla}
        \label{fig:result_init_time_sub_gowalla}
     \end{subfigure}
     \hfill
     \begin{subfigure}{\textwidth}
         \centering
         \includegraphics[width=\textwidth]{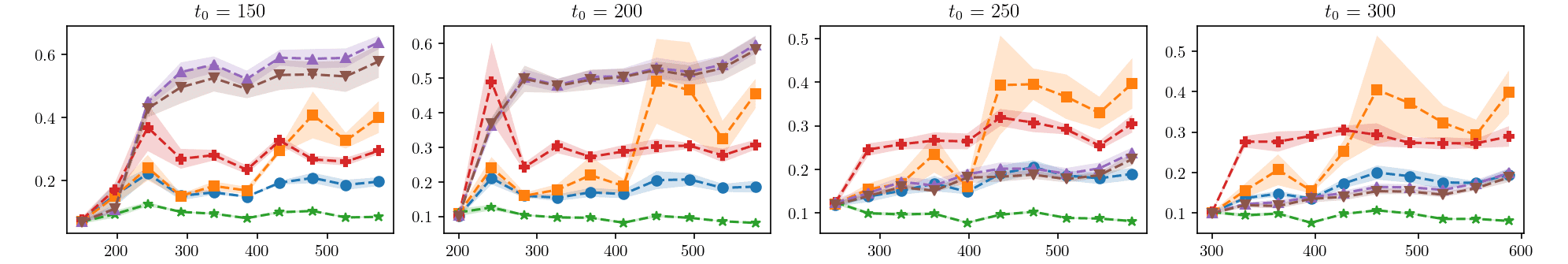}
        \caption{Gowalla with deletion}
        \label{fig:result_init_time_sub_gowalla_with_deletion}
     \end{subfigure}
     \hfill
     \begin{subfigure}{\textwidth}
         \centering
         \includegraphics[width=\textwidth]{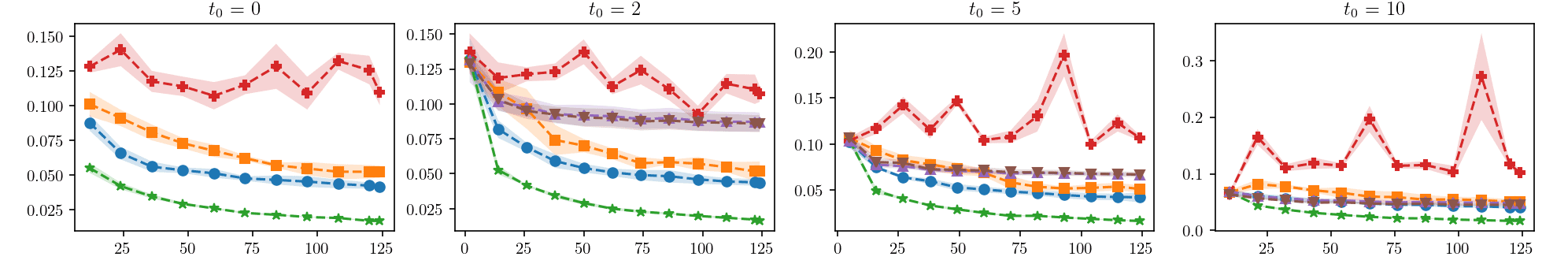}
        \caption{New York taxi (over NY State)}
        \label{fig:result_init_time_sub_ny_state}
     \end{subfigure}
     \hfill
     \begin{subfigure}{\textwidth}
         \centering
         \includegraphics[width=\textwidth]{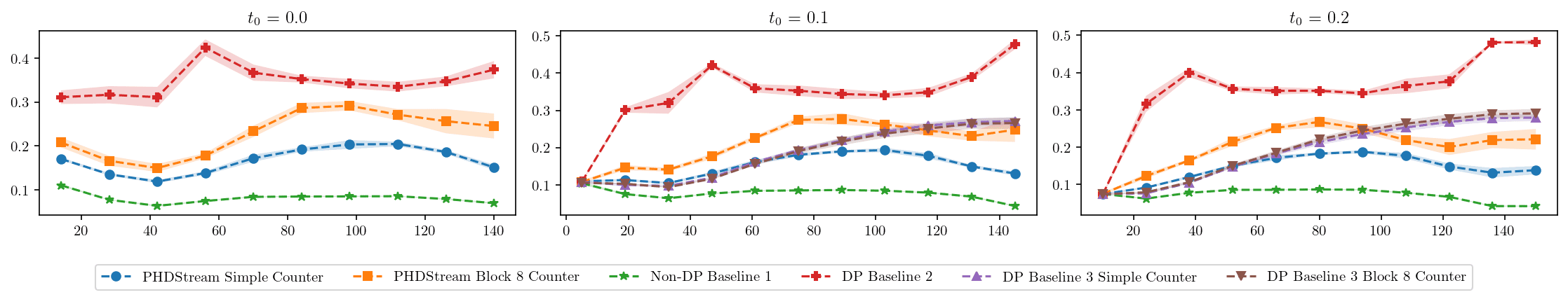}
        \caption{Concentric circles with deletion}
        \label{fig:result_init_time_sub_circles_motion_deletion}
     \end{subfigure}
    \caption{The progression of relative error in small range queries with time. All experiments are with privacy budget $\e=0.5$. Each subplot has a time horizon on the x-axis and corresponds to a particular value of $t_0$ (increasing from left to right).}
    \label{fig:result_init_time}
    \vskip -0.1in
\end{figure*}

\subsection{Datasets}
We conducted experiments on datasets with the location (latitude and longitude coordinates) of users collected over time. We analyzed the performance of our method on two real-world and three artificially constructed datasets.

\textbf{Real-world Datasets:} We use two real-world datasets: Gowalla \cite{gowalla} and NY Taxi \cite{ny_data}. The Gowalla dataset contains location check-ins by a user on the social media app Gowalla. The NY Taxi dataset \cite{ny_data} contains the pickup and drop-off locations and times of Yellow Taxis in NY.

\textbf{Simulated Datasets:} Our first simulated dataset is derived from the Gowalla dataset. Motivated by the problem of releasing the location of active coronavirus infection cases, we consider each check-in as a report of infection and remove this check-in after $30$ days indicating that the reported individual has recovered. This creates a version of Gowalla with deletion. The second dataset contains points sampled on two concentric circles but the points first appear on one circle and then gradually move to the other circle. This creates a dataset where the underlying distribution changes dramatically. We also wanted to explore the performance of our algorithm based on the number of data points it observes in initialization and then at each batch. Thus for this dataset, we first fix a constant batch size for the experiment and then index time accordingly. Additionally, we keep initialization time $t_0$ as a value in $[0,1]$ denoting the proportion of total data the algorithm uses for initialization.

\subsection{Initialization}
Many real-world datasets such as Gowalla have a slow growth at the beginning of time. In practice, we can hold the processing of such streams until some time $t_0\in\N$, which we refer to as initialization time, until sufficient data has been collected. Thus for any input data stream $f:\Omega\times\N\to\R$ and initialization time $t_0\in\N$, we use a modified stream $\hat f: \Omega\times\N\to\R$ such that $\hat f(\cdot, t) = f(\cdot, t+t_0-1)$ for all $t\in\N$. We discuss the effects of initialization time $t_0$ on the algorithm performance in detail in Section~\ref{s: results}

\subsection{Baselines}
Let $\e$ be our privacy budget. Consider the offline synthetic data generation algorithm as per \cite{privtree}, which is to use $\PrivTree_T$ with privacy budget $\e/2$ to obtain a subtree $S$ of $T$, generate the synthetic count of each node $v\in \LL(S)$ by adding Laplace noise with scale $2/\e$, and sample points in each node $v \in \LL(S)$ equal to their synthetic count. Let us denote this algorithm as $\PrivTree_T+\Counting$. A basic approach to convert an offline algorithm to a streaming algorithm is to simply run independent instances of the algorithm at each time $t \in \N$. We use this idea to create the following three baselines and compare our algorithm with them.

\textbf{Baseline 1) Offline PrivTree on stream:} At any time $t\in\N$, we run an independent version of the offline $\PrivTree_T+\Counting$ on $f(\cdot, t)$, that is the stream data at time $t$. We want to emphasize that to be differentially private, this baseline requires a privacy parameter that scales with time $t$ and it does NOT satisfy differential privacy with the parameter $\e$. Thus we expect it to perform much better with an algorithm that is $\e$-DP. We still use this method as a baseline since, due to a lack of DP algorithms for streaming data, industry applications sometimes naively re-run DP algorithms as more data is collected. However, we show that PHDStream performs competitively close to this baseline.

\textbf{Baseline 2) Offline PrivTree on differential stream} Similar to Baseline~1, at any time $t\in\N$, we run an independent version of the offline $\PrivTree_T+\Counting$ algorithm, but with the input data $\nabla f(\cdot, t)$, that is the differential data observed at $t$. This baseline satisfies $\e$-DP. For a fair comparison, at time $t$, any previous output is discarded and the current synthetic data is scaled by the factor ${|f(\cdot, t)|}/{|\nabla f(\cdot, t)|}$.

\textbf{Baseline 3) Initialization with PrivTree, followed by counting with counters:} We first use the offline $\PrivTree_T+\Counting$ algorithm at only the initialization time $t_0$ and get a subtree $S$ of $T$ as selected by PrivTree. We then create a counter for all nodes $v\in \LL(S)$ with privacy budget $\e$. At any time $t>t_0$, $S$ is not updated, and only the counter for each of the nodes in $\LL(S)$ is updated using the differential stream $\nabla f(\cdot, t)$. This baseline also satisfies $\e$-DP. Note that this algorithm has twice the privacy budget for counting at each time $t>t_0$ as we do not update $S$. We only have results for this baseline if $t_0>0$. We show that PHDStream outperforms this baseline if the underlying density of points changes sufficiently with time.

\subsection{Performance metric}\label{s: metric}
We evaluate the performance of our algorithm against {\em range counting queries}, that count the number of points in a subset of the space $\Omega$. For $\Omega'\subseteq\Omega$, the associated range counting query $q_{\Omega'}$ over a function $h:\Omega \to \R$ is defined as $q_{\Omega'}(h) \coloneqq \sum_{x\in \Omega'}h(x)$. In our experiments, we use random rectangular subsets of $\Omega$ as the subregion for a query. Similar to \cite{privtree}, we generate three sets of queries: $10,000$ small, $5000$ medium, and $1000$ large with area in $[0.01\%, 0.1\%)$, $[0.1\%, 1\%)$, and $[1\%, 10\%)$ of the space $\Omega$ respectively. We use average relative error over a query set as our metric. At time $t$, given query set $Q$, the relative error of a synthetic stream $g:\Omega\times\N\to\R$ as compared to the input stream $f:\Omega\times\N\to\R$ is thus defined as, $$r(Q, f, g, t) \coloneqq \frac{1}{|Q|}\sum_{q\in Q} \frac{\abs{q(f(\cdot, t))-q(g(\cdot, t))}}{\max{\inparanth{q(f(\cdot, t)), 0.001|f(\cdot, t)| }}}$$
where $0.001|f(\cdot, t)|$ or $1\%$ of $|f(\cdot, t)|$ is a small additive term to avoid division by $0$. We evaluate the metric for each query set at a regular time interval and report our findings.

\subsection{Results}\label{s: results}

Our analysis indicates that PHDStream performs well across various datasets and hyper-parameter settings. We include some of the true and synthetic data images in Appendix~\ref{app: synthetic scatter plots}. Due to space constraints, we limit the discussion of the result to a particular setting in Figure~\ref{fig:result_init_time} where we fix the privacy budget to $\e=1$, sensitivity to $2$, query set to small queries, and explore different values for initialization time $t_0$. For further discussion, refer to Appendices~\ref{app: synthetic scatter plots} and~\ref{app: more results}.

PHDStream remains competitive to the challenging non-differentially private Baseline~1 and in almost all cases, it outperforms the differentially private Baseline~2. It also outperforms Baseline~3 if we initialize the algorithm sufficiently early. We discuss above mentioned observations together with the effect of various hyper-parameters below.

\textbf{Initialization time:} If the dataset grows without changing the underlying distribution too much, it seems redundant to update $T(t)$ with $\PrivTree_T$ at each time $t \in \N$. Moreover, when counting using fixed counters, we know that error in a counter grows with time so the performance of the overall algorithm should decrease with time. We observe the same for higher values of initialization time in Figure~\ref{fig:result_init_time}. Moreover, we see that Baseline~3, which uses PrivTree only once, outperforms PHDStream for such cases.

\textbf{Counter type:}
In almost all cases, for the PHDStream algorithm, the Simple counter has better performance than the Block 8 counter. This can be explained by the fact that for these datasets, the majority of nodes had their counter activated only a few times in the entire algorithm. For example, when using the Gowalla dataset with $t_0=100$, on average, at least $90\%$ of total nodes created in the entire run of the algorithm had their counter activated less than $40$ times.


\textbf{Sensitivity and batch size:} We explore various values of sensitivity and batch size and the results are as expected - low sensitivity and large batch size improve the algorithm performance. For more details see Appendix~\ref{app: more results}.

%% file: 6_appendix.tex
\appendix
\onecolumn


\section{$\PrivTree_T$ algorithm and its privacy} \label{app: PrivTree proof}

\begin{algorithm}[ht]
    \begin{algorithmic}[1]
        \STATE {\bfseries Input:} $F \in C(T)$, the privacy budget parameter $\e$, the threshold count for a node $\theta$.
        \STATE {\bfseries Output:} A subset $S$ of the tree $T$.
        \STATE Set $\lambda \leftarrow \frac{2\beta -1}{\beta -1}\cdot \frac{2}{\e}$ and $\delta \leftarrow \lambda \ln{\beta}$.
        \STATE Initialize an empty tree $S$.
        \STATE $U \leftarrow \set{root(T)}$; denotes a set of un-visited nodes.
        \FOR{$v \in U$}
            \STATE Add $v$ to $S$.
            \STATE $U \leftarrow U\backslash\set{v}$; mark $v$ visited.
            \STATE $b(v) \leftarrow F(v) - depth(v) \cdot \delta$; is a biased aggregate count of the node $v$.
            \STATE $b(v) \leftarrow \max \set{b(v), \theta-\delta}$; to ensure that $b(v)$ is not too small.
            \STATE $\tilde b(v) \leftarrow b(v) + Lap(\lambda)$; noisy version of the biased count.
            \IF{$(\tilde b(v) >\theta)$; noisy biased count is more than threshold}
                \STATE Add children of $v$ in $T$ to $U$.
            \ENDIF
        \ENDFOR        
    \end{algorithmic}
    \caption{$\PrivTree_T$}
    \label{alg: PrivTree}
\end{algorithm}

\begin{proof}[Proof of Theorem~\ref{thm: PrivTree privacy}]
    Similar to \cite{privtree} we define the following two functions:
    \begin{equation}\label{eq:rho}
        \rho_{\lambda, \theta}(x) = \ln{\inparanth{ \frac{ \prob{x+Lap(\lambda)>\theta } } { \prob{x-1+Lap(\lambda)>\theta} } } },
    \end{equation}
    \begin{equation}\label{eq:rho_top}
        \rho^\top_{\lambda, \theta}(x) = \begin{cases}
            1/\lambda, & x<\theta+1, \\
            \frac{1}{\lambda} \exp{ \inparanth{ \frac{\theta+1-x}{\lambda} } }, & otherwise.
        \end{cases}
    \end{equation}
    It can be shown that $\rho_{\lambda, \theta}(x) \leq \rho^\top_{\lambda, \theta}(x)$ for any $x$.
    
    Let us consider the neighbouring functions $F$ and $\tilde F$ such that $\norm{F-\tilde F}_{C(T)}=1$ and there exists $v_*\in V(T)$ such that $F(v_*) \neq \tilde F(v_*)$. Consider the output $S$, a subtree of $T$.
    
    \textit{Case 1: } $F(v_*)>\tilde F(v_*)$
    
    Note that there will be exactly one leaf node in $S$ that differs in the count for the functions $F$ and $\tilde F$. Let $v_1, v_2, \dots, v_k$ be the path from the root to this leaf node. As per Algorithm~\ref{alg: PrivTree}, let us denote the calculated biased counts of any node $v \in S$ for the functions $F$ and $\tilde F$ as $b(v)$ and $\tilde b(v)$, respectively. Then, for any $v \in S$, we have the following relations,
    \begin{align*}
        \tilde F(v) &= \begin{cases}
                    F(v) - 1, & v\in \set{v_1, v_2, v_3, \dots, v_k} \\
                    F(v), & otherwise,
                \end{cases} \\
        \tilde b(v) &= \begin{cases}
                    b(v) - 1, & v\in \set{v_1, v_2, v_3, \dots, v_k}, b(v) > \theta-\delta, \\
                    b(v), & otherwise.
                \end{cases}
    \end{align*}
    Let there exist an $m \in [k-1]$, s.t. $b(v_m) \geq \theta-\delta+1$ and $b(v_{m+1})=\theta-\delta$. We then show that there is a difference in count by at least $\delta$ between parent and child for all $i \in [2,m]$
    \[
        b(v_i) = F(v_i) - \delta \cdot depth(v)
        \leq F(v_{i-1}) - \delta\cdot (depth(v_{i-1})+1)
        \leq b(v_{i-1})-\delta .
    \]
    Thus we have,
    \begin{equation} \label{eq:countincrement}
        \begin{cases}
            b(v_{i-1}) \geq b(v_i)+\delta \geq \theta +1, & i\in[2,m], \\
            b(v_i) =\theta - \delta, & otherwise.
        \end{cases}
    \end{equation}
    Finally, we can show DP as follows,
    \begin{align*}
        \ln{\inparanth{\frac{\prob{\MM(F) = S}}{\prob{\MM(\tilde F) = S}}}} 
        = {}& \sum_{i=1}^{k-1} \ln{\inparanth{\frac{\prob{b(v_i) + Lap(\lambda) > \theta}}{\prob{\tilde b(v_i) + Lap(\lambda) > \theta}}}}
        + \ln{\inparanth{\frac{\prob{b(v_k) + Lap(\lambda) \leq \theta}}{\prob{\tilde b(v_k) + Lap(\lambda) \leq \theta}}}}
         \\
        = {} & \sum_{i=2}^{m} \ln{\inparanth{\frac{\prob{b(v_i) + Lap(\lambda) > \theta}}{\prob{b(v_i) - 1 + Lap(\lambda) > \theta}}}} 
        + \sum_{i=m+1}^{k-1} \ln{\inparanth{\frac{\prob{b(v_i) + Lap(\lambda) > \theta}}{\prob{b(v_i) + Lap(\lambda) > \theta}}}} 
        \\
        &+ \ln{\inparanth{\frac{\prob{b(v_k) + Lap(\lambda) \leq \theta}}{\prob{ \tilde b(v_k) +Lap(\lambda) \leq \theta}}}} \\
        & [\text{Using the fact that } \ln{x} \leq 0 \text{ for all } x \leq 1] \\
        \leq {} & \sum_{i=2}^{m} \ln{\inparanth{\frac{\prob{b(v_i) + Lap(\lambda) > \theta}}{\prob{b(v_i) - 1 + Lap(\lambda) > \theta}}}} 
        + 0 +0.
    \end{align*}
    Using equations \eqref{eq:rho} and \eqref{eq:rho_top} we have,
    \begin{align*}
        \ln{\inparanth{\frac{\prob{\MM(F) = S}}{\prob{\MM(\tilde F) = S}}}}
        \leq {} & \sum_{i=2}^{m}  \rho_{\lambda, \theta}{(b(v_i))} 
        \leq \sum_{i=2}^{m}  \rho_{\lambda, \theta}^\top{(b(v_i))}
        =  \rho_{\lambda, \theta}^\top{(b(v_m))} + \sum_{i=2}^{m-1}  \rho_{\lambda, \theta}^\top{(b(v_i))} \\
        = {} & \frac{1}{\lambda} + \sum_{i=2}^{m-1}  \frac{1}{\lambda} \exp{\inparanth{\frac{\theta+1-b(v_i)}{\lambda}}}
        = \frac{1}{\lambda} + \sum_{i=1}^{m-2}  \frac{1}{\lambda} \exp{\inparanth{\frac{\theta+1-b(v_{m-i})}{\lambda}}} \\
        & [\text{Using equation \eqref{eq:countincrement} and } b(v_{m-1}) \geq b(v_m) + \delta \geq \theta + 1] \\
        \leq {} & \frac{1}{\lambda} + \frac{1}{\lambda} \sum_{i=1}^{m-2}   e^{-\delta/\lambda}
        \leq \frac{1}{\lambda} + \frac{1}{\lambda} \cdot \frac{1}{1-e^{-\delta/\lambda}}
        =  \frac{1}{\lambda} \cdot \frac{2e^{\delta/\lambda}-1}{e^{\delta/\lambda}-1} = \e.
    \end{align*}

    \textit{Case 2: } $F(v_*)<\tilde F(v_*)$
    
    Following the same notations as Case 1, for any $v \in S$, we have the following relations,
    \begin{align*}
        \tilde F(v) &= \begin{cases}
                    F(v) + 1, & v\in \set{v_1, v_2, v_3, \dots, v_k}, \\
                    F(v), & otherwise,
                \end{cases} \\
        \tilde b(v) &= \begin{cases}
                    b(v) + 1, & v\in \set{v_1, v_2, v_3, \dots, v_k}, b(v) \geq \theta-\delta, \\
                    b(v), & otherwise.
                \end{cases}
    \end{align*}
    Finally, we can show DP as follows,
    \begin{align*}
        \ln{\inparanth{\frac{\prob{\MM(F) = S}}{\prob{\MM(\tilde F) = S}}}} 
        = {}& \sum_{i=1}^{k-1} \ln{\inparanth{\frac{\prob{b(v_i) + Lap(\lambda) > \theta}}{\prob{\tilde b(v_i) + Lap(\lambda) > \theta}}}}
        + \ln{\inparanth{\frac{\prob{b(v_k) + Lap(\lambda) \leq \theta}}{\prob{\tilde b(v_k) + Lap(\lambda) \leq \theta}}}}
         \\
        & [\text{Since } \ln{x} \leq 0 \text{ for all } x \leq 1] \\
        &\leq 0 + \ln{\inparanth{\frac{\prob{b(v_k) + Lap(\lambda) \leq \theta}}{\prob{\tilde b(v_k) + 1 + Lap(\lambda) \leq \theta}}}}
         \leq \e.
    \end{align*}

    Hence the algorithm is $\e$-differentially private.
    
\end{proof}


\section{Poof of privacy of PHDStream}\label{app: PHDStream proof}
\begin{proof}[Proof of Theorem~\ref{thm: PHDStream privacy}]
    It suffices to show that PHDStreamTree is $\e$-differentially private, since PHDStream is doing a post-processing on its output which is independent of the input data stream. Let $f$ and $\tilde f$ be neighboring data streams such that $\norm{f-\tilde f}_\nabla=1$. Let $F$ and $\tilde F$ be the corresponding streams on the vertices of $T$ respectively. Since $f$ and $\tilde f$ are neighbors, there exists $\tau \in \N$ such that $\norm{\nabla F(\cdot, \tau) - \nabla \tilde F (\cdot, \tau)}_{C(T)}=1$. Note that, at any time $t$, PHDStreamTree only depends on $\nabla F_t$ from the input data stream. Hence for the differential privacy over the entire stream, it is sufficient to show that the processing at time $\tau$ is differentially private. Note that both $\PrivTree_T$ and the calculation of $d(\cdot, t)$ satisfy $\e/2$-differential privacy by Theorem~\ref{thm: PrivTree privacy} and by the Laplace Mechanism, respectively. Hence their combination, that is PHDStreamTree at time $\tau$, satisfies $\e$-differential privacy. At any $t \neq \tau$, $\nabla F(v, t)=\nabla \tilde F(v,t)$ for all $v \in T$. Hence PHDStreamTree satisfies $\e$-differential privacy for the entire input stream $F$.
\end{proof}

\input{appendices/optimizing_algo}

\input{appendices/counters}


\section{Privacy of selective counting}\label{app: proof selective counting}
\begin{proof}
    Let us denote Algorithm~\ref{alg:selective_counting} as $\AA$. Let $f$ and $\tilde f$ be neighboring data streams with $\norm{f-\tilde f}=1$. Without loss of generality, there exists $\tau\in\N$ such that $f(\tau)=\tilde f(\tau)+1$. Let $g$ be the output stream we are interested in. Let $\restr{f}{[t]}$ denote a restriction of the stream $f$ until time $t$ for any $t \in \N$.
    
    Since the input streams are identical till time $\tau$, we have, 
    \begin{align*}
        \prob{\AA \inparanth{\restr{f}{[\tau-1]}} = \restr{g}{[\tau-1]}}
        &= \prob{\AA\inparanth{\restr{\tilde f}{[\tau-1]}} = \restr{g}{[\tau-1]}}
    \end{align*}
    At time $\tau$, let $g(\tau) = (l_\tau, m_\tau)$. Since $\MM$ is $\e_M$-DP, we have,
    \begin{align*}
        \prob{\MM \inparanth{ \restr{f}{[\tau]}, \restr{g}{[\tau-1]} } = l_\tau}
        &\leq e^{\e_M} \cdot \prob{\MM\inparanth{\restr{\tilde f}{[\tau]}, \restr{g}{[\tau-1]}} = l_\tau}.
    \end{align*}
    Moreover, since $\MM_{l_\tau}$ is $\e_C$-DP,
    \begin{align*}
        \prob{\CC_{l_\tau} \inparanth{ \restr{f}{J_{l_t}} } = m_\tau}
        &\leq e^{\e_C} \cdot \prob{\CC_{l_\tau}\inparanth{\restr{\tilde f}{J_{l_t}}} = m_\tau}.
    \end{align*}
    Combining the above two we have,
    \begin{align*} 
        \prob{\AA \inparanth{\restr{f}{[\tau]} } = g(\tau)}
        &= \prob{\AA \inparanth{ \restr{f}{[\tau-1]} } = \restr{g}{[\tau-1]}} \cdot 
        \prob{\MM \inparanth{ \restr{f}{[\tau]}, \restr{g}{[\tau-1]} } = l_\tau} \cdot
        \prob{\CC_{l_\tau} \inparanth{ \restr{f}{J_{l_t}} } = m_\tau}
        \\
        &\leq \prob{\AA \inparanth{\restr{\tilde f}{[\tau-1]}} = \restr{g}{[\tau-1]}} \cdot 
        e^{\e_M} \prob{\MM \inparanth{\restr{\tilde f}{[\tau]}, \restr{g}{[\tau-1]}} = l_\tau}
        \cdot e^{\e_C} \prob{\CC_{l_\tau} \inparanth{ \restr{\tilde f}{J_{l_t}} } = m_\tau}
        \\
        &= e^{(\e_M+\e_C)} \cdot \prob{\AA\inparanth{\restr{\tilde f}{[\tau]}} = g(\tau)}.
    \end{align*}
    At any time $t>\tau$, since the input data streams are identical, we have,
    \begin{align*}
        \prob{\AA \inparanth{\restr{f}{[t]} } = \restr{g}{[t]} }
        &\leq e^{(\e_M+\e_C)} \cdot \prob{\AA \inparanth{ \restr{\tilde f}{[t]} } = \restr{g}{[t]} }.
    \end{align*}
    Hence, algorithm $\AA$ is $(\e_M+\e_C)$-DP.    
\end{proof}

\section{Other experiment details}
All the experiments were performed on a personal device with 8 cores of 3.60GHz 11th Generation Intel Core i7 processor and 32 GB RAM. Given the compute restriction, for any of the datasets mentioned below, we limit the total number of data points across the entire time horizon to the (fairly large) order of $10^5$. Moreover, the maximum time horizon for a stream we have is $300$ for Gowalla Dataset. We consider it to be sufficient to show the applicability of our algorithm for various use cases.

We implement all algorithms in Python and in order to use the natural geometry of the data domain we rely on the spatial computing libraries GeoPandas \cite{geopandas_kelsey_jordahl_2020_3946761}, Shapely \cite{Gillies_Shapely_2023}, NetworkX \cite{networkx_SciPyProceedings_11}, and OSMnx \cite{osmnx_Boeing2017}. These libraries allow us to efficiently perform operations such as loading the geometry of a particular known region, filtering points in a geometry, and sampling/interpolating points in a geometry.


\input{appendices/datasets}

\input{appendices/synthetic_data}

\input{appendices/more_results}

%% file: appendices/optimizing_algo.tex
\section{Optimizing the algorithm}\label{app:optim}

In Algorithm~\ref{alg: PHDStreamTree}, at any time $t$, we have to perform computations for every node $v \in T$. Since $T$ can be a tree with large depth and the number of nodes is exponential with respect to depth, the memory and time complexity of PHDStreamTree, as described in Algorithm~\ref{alg: PHDStreamTree} is also exponential in depth of the tree $T$. Note however that we do not need the complete tree $T$, and can limit all computations to the subtree $T(t)$ as selected by PrivTree at time $t$. Based on this idea, we propose a version of PHDStream in Algorithm~\ref{alg: Compute efficient PHDStream} which is storage and runtime efficient.

\begin{algorithm}[ht]
    \begin{algorithmic}[1]
        \STATE {\bfseries Input:} Input data stream $f: \Omega \times \N \to \R$, a fixed and known tree $T$, the privacy budget parameter $\e$, the threshold count for a node $\theta$.
        \STATE {\bfseries Output:} A synthetic data stream $g: \Omega \times \N \to \R$.
        \STATE Set $\lambda \leftarrow \frac{2\beta -1}{\beta -1}\cdot \frac{2}{\e}$ and $\delta \leftarrow \lambda \ln{\beta}$.
        \STATE Initialize $C_a$, $C_n$, and $C_d$ and make them accessible to Algorithm~\ref{alg: Modified PrivTree}.
        \FOR{every time $t \in \N$}
            \STATE Set $T(t)$ as the output of Algorithm \ref{alg: Modified PrivTree} with parameters
            $(\nabla f(\cdot, t), \delta, \theta, \lambda, \e/2)$.
            \STATE Calculate $G(v) \leftarrow C_a(v) + C_n(v) + C_d(v)$ for all $v \in leaves(T(t))$.
            \STATE Convert $G$ to $g(\cdot, t)$ as per equation~\ref{eq: contraction reversed}.
        \ENDFOR        
    \end{algorithmic}
    \caption{Compute efficient PHDStream}
    \label{alg: Compute efficient PHDStream}
\end{algorithm}

\begin{algorithm}[h]
    \begin{algorithmic}[1]
        \STATE {\bfseries Input:} Count of points in $\Omega$ as $h: \Omega \to \R$, parameters related to privacy $\delta, \theta, \lambda, \e'$.
        \STATE {\bfseries Output:} A subtree $T^*$ of $T$.
        \STATE Initialize an empty subtree $T^*$ of $T$.
        \STATE Initialize a queue $Q$ of un-visited nodes.
        \STATE Push $root(T)$ to $Q$.
        \STATE Initialize an empty stack $S$.
        \WHILE{$Q$ is not empty}
            \STATE Pop $v$ from $Q$.
            \STATE Add $v$ to the subtree $T^*$.
            \STATE Find $H(v)$ using $h$ as per equation~\ref{eq: contraction}.
            \STATE To enforce ancestor consistency, update\\ $C_a(v) \leftarrow \frac{1}{\beta}\cdot \inparanth{ C_a(parent(v)) + C_n(parent(v)) }$.
            \STATE Set $s(v) \leftarrow C_a(v) + C_n(v) + C_d(v)$
            \STATE Set $b(v) \leftarrow \inparanth{s(v) + H(v) - depth(v) \cdot \delta}$; as biased count of the node $v$.
            \STATE $b(v)\leftarrow\max\set{b(v), \theta-\delta}$; to ensure that $b(v)$ is not too small.
            \STATE $\hat b(v) \leftarrow b(v) + Lap(\lambda)$; noisy version of the biased count.
            \IF{$\inparanth{\hat b(v) >\theta}$ and $children(v)$ is not empty in $T$}
                \STATE Push $w$ to $Q$ for all $w \in children(v)$ in $T$.
                \STATE Push $v$ to $S$.
            \ELSE
                \STATE Update $C_n(v) \leftarrow C_n(v) + H(v) + Lap(2/\e')$ \label{lst: Counter Update in compute efficient PHDStream}.
            \ENDIF        
        \ENDWHILE
        \WHILE{$S$ is not empty}
            \STATE Pop $v$ from $S$.
            \STATE To enforce descendant consistency, update\\ $C_d(v) \leftarrow \sum_{w \in children(v)} \inparanth{C_d(w) + C_n(w)}$.
        \ENDWHILE        
    \end{algorithmic}
    \caption{Modified $\PrivTree_T$ subroutine}
    \label{alg: Modified PrivTree}
\end{algorithm}

Note that in Algorithm~\ref{alg: PHDStreamTree}, at any time $t\in\N$, we find the differential synthetic count of any node $v$ as $d(v,t)$ only if $v$ is a leaf in $T(t)$. To maintain the consistency, as given by equation~\ref{eq: consistency}, this count gets pushed up to ancestors or spreads down to the descendants by the consistency extension operator $Ext_T$. A key challenge in efficiently enforcing consistency is that we do not want to process nodes that are not present in $T(t)$. Since the operator $Ext_T$ is linear if we have the total differential synthetic count of all nodes $v$ till time $t$, that is $\sum_{t' \leq t} d(v,t')$, we can find the count of any node resulting after consistency extension. We track this count in a function $C_n : V(T) \to \R$ such that at any time $t$, $C_n(v)$ represents the total differential synthetic count of node $v$ for times $t' \leq t$ when $v \in leaves(T(t'))$. We update the value $C_n(v)$ for any node $v \in leaves(T(t))$ as,
$$
C_n(v) \leftarrow C_n(v) + \inparanth{\nabla F(v)+Lap(2/\e)}.
$$

Furthermore, to efficiently enforce consistency, we introduce $C_a$ and $C_d$ mappings $V(T) \to \R$. At any time $t$, $C_a(v)$ and $C_d(v)$ denote the count a node $v$ has received at time $t$ when enforcing consistency from its ancestors and descendants, respectively. Consistency due to counts being passed down from ancestors can be enforced with the equation
\begin{equation}\label{eq: ancestor consistency}
    C_a(v) = \frac{1}{\beta}\cdot \inparanth{ C_a(parent(v)) + C_n(parent(v)) }.
\end{equation}
Consistency due to counts being pushed up from descendants can be enforced with the equation
\begin{equation}\label{eq: descendant consistency}
    C_d(v) = \sum_{w \in children(v)} \inparanth{C_d(w) + C_n(w)}.
\end{equation}

We assume that $C_a(v)=C_n(v)=C_d(v)=0$ for all $v \in V(T)$ at the beginning of Algorithm~\ref{alg: Compute efficient PHDStream}. With the help of these functions, the synthetic count of any node $v \in T(t)$ can be calculated as $s(v) = C_a(v)+C_n(v)+C_d(v)$. 

Algorithm~\ref{alg: Modified PrivTree} is a modified version of PrivTree Algorithm~\ref{alg: PrivTree}. We use it as a subroutine of Algorithm~\ref{alg: Compute efficient PHDStream}, where at time $t$ it is responsible for creating the subtree $T(t)$ and updating the values of the functions $C_a$, $C_n$, and $C_d$.

The algorithm uses two standard data structures called \textit{Stack} and \textit{Queue} which are ordered sets where the order is based on the time at which the elements are inserted. We interact with the data structures using two operations \textit{Push} and \textit{Pop}. Push is used to insert an element, whereas pop is used to remove an element. The key difference between a Queue and a Stack is that the operation pop on a Queue returns the element inserted earliest but on a Stack returns the element inserted latest. Thus Queue and Stack follow the FIFO (first-in, first-out) and LIFO (last-in, first-out) strategy, respectively.

\subsection{Privacy of compute-efficient PHDStream}
At any time $t\in\N$, the key difference of Algorithm~\ref{alg: Compute efficient PHDStream} (compute efficient PHDStream) from Algorithm~\ref{alg: PHDStream} (PHDStream) is that (1) it does not compute $\nabla F(v,t)$ for any node $v \in T\setminus T(t)$, and (2) it does not use the consistency operator $Ext_T$ and instead relies on the mappings $C_a$, $C_n$, and $C_d$. Note that none of these changes affect the constructed tree $T(t)$ and the output stream $g(\cdot, t)$. Since both algorithms are equivalent in their output and Algorithm~\ref{alg: PHDStream} is $\e$-differentially private, we conclude that Algorithm~\ref{alg: Compute efficient PHDStream} is also $\e$-differentially private.

%% file: appendices/counters.tex
\section{Counters}
\label{app:counters}

We provide the algorithms for Simple, Block, and Binary Tree counter in Algorithms~\ref{alg:appendix_simple_counter},~\ref{alg:appendix_block_simple_counter}, and~\ref{alg:appendix_bt_counter} respectively. As proved in \cite{Dwork2010ContinualDP, Chan2010ContinualPrivateStats} for a fixed failure probability $\d$, ignoring the constants, the accuracy $\a$ at time $t$ for the counters Simple, Block, and Binary Tree is $\bigo{\sqrt{t}}$, $\bigo{t^{1/4}}$, and $\bigo{\inparanth{\ln{t}}^{3/2}}$ respectively. Hence the error in the estimated value in counters grows with time $t$. The result also suggests that for large time horizons, the Binary Tree algorithm is best to be used as a counter. However, for small values of time $t$, the bounds suggest that Simple and Block counters are perhaps more effective. Note that the optimal size of the block in a Block Counter is suggested to be $\lfloor \sqrt{t} \rfloor$ where $t$ is the time horizon for the stream. However, in our method, we do not know in advance the time horizon for the stream that is input for a particular counter. After multiple experiments, we found that a Block counter with size 8 was most effective for our method PHDStream for the datasets in our experiments. Figure~\ref{fig:compare_counters} shows an experiment comparing these counters. We observe that the Binary Tree algorithm should be avoided on streams of small time horizons, say $T<100$. Block counter appears to have lower errors for $T>30$. However, due to the large growth in error within a block, for $T<30$ it is unclear if the Block counter is better than the Simple counter.

\begin{algorithm}[h]
    \begin{algorithmic}[1]
        \STATE \textbf{Input:} Input stream $f:\N\to\R$, privacy budget $\e$
        \STATE \textbf{Output:}  Output stream $g:\N\to\R$
        \STATE Assume $g(0) \leftarrow 0$
        \STATE Set $g(t) \leftarrow g(t-1)+f(t) + Lap(\frac{1}{\e})$ for all $t\in\N$
    \end{algorithmic}
    \caption{Simple counter}
    \label{alg:appendix_simple_counter}
\end{algorithm}

\begin{algorithm}[h]
    \begin{algorithmic}[1]
        \STATE \textbf{Input:} Input stream $f:\N\to\R$, privacy budget $\e$, block size $B$
        \STATE \textbf{Output:}  Output stream $g:\N\to\R$
        \STATE $\a \leftarrow 0, \b \leftarrow 0, \b_{lastblock} \leftarrow 0$
        \FOR{$t=1, 2, \dots$}
            \STATE $\b \leftarrow \b + f(t)$
            \IF{$t=k\cdot B$, for some $k\in\N$}
                \STATE $\a \leftarrow 0$, $\b \leftarrow \b + Lap(\frac{2}{\e})$
                \STATE $\b_{lastblock} \leftarrow \b$
            \ELSE
                \STATE $\a \leftarrow \a + f(t) + Lap(\frac{2}{\e})$
            \ENDIF
            \STATE $g(t) \leftarrow \b_{lastblock} + \a$
        \ENDFOR
    \end{algorithmic}
    \caption{Block counter with block size $B$}
    \label{alg:appendix_block_simple_counter}
\end{algorithm}

\begin{algorithm}[h]
    \begin{algorithmic}[1]
        \STATE \textbf{Input:} Input stream $f:\N\to\R$, privacy budget $\e$, time horizon $T$
        \STATE \textbf{Output:}  Output stream $g:\N\to\R$
        \STATE Initialize $\a_i, \hat \a_i \leftarrow 0$ for all $i\in \N$
        \FOR{$t=1, 2, \dots$} 
            \STATE Let $b_n\dots b_1b_0$ be the $(n+1)$-bit binary representation of $t$, i.e., $t=\sum_{i=0}^{n} b_i2^i$
            \STATE $j \leftarrow \min{\set{i: b_i \neq 0}}$
            \STATE $\a_j \leftarrow \inparanth{\sum_{i=0}^{j-1} \a_i} + f(t)$
            \STATE $\hat \a_j \leftarrow \a_j + Lap(\frac{\log_2{T}}{\e})$
            \STATE $\a_i, \hat \a_i \leftarrow 0$ for all $i<j$
            \STATE $g(t) \leftarrow \sum_{i=0}^{n} \hat \a_i \cdot (\ind_{b_i=1})$
        \ENDFOR
    \end{algorithmic}
    \caption{Binary Tree counter}
    \label{alg:appendix_bt_counter}
\end{algorithm}

\begin{figure}[h]
    \vskip 0.1in
    \centering
    \begin{subfigure}[b]{0.45\columnwidth}
         \centering
         \includegraphics[width=\linewidth]{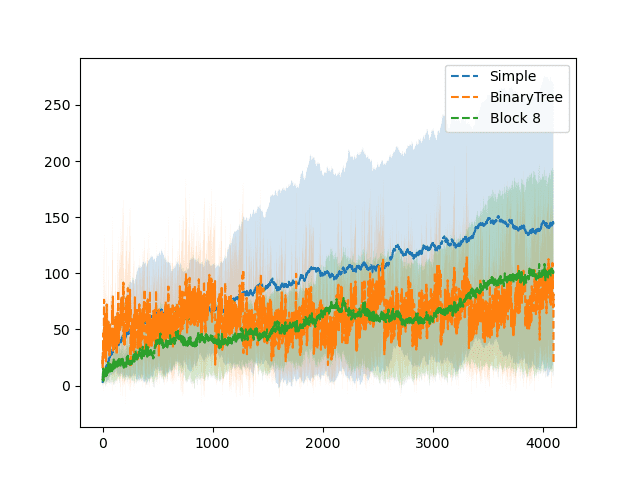}
         \caption{Long time horizon of $2^{12}=4096$}
     \end{subfigure}
     \hfill
     \begin{subfigure}[b]{0.45\columnwidth}
         \centering
         \includegraphics[width=\linewidth]{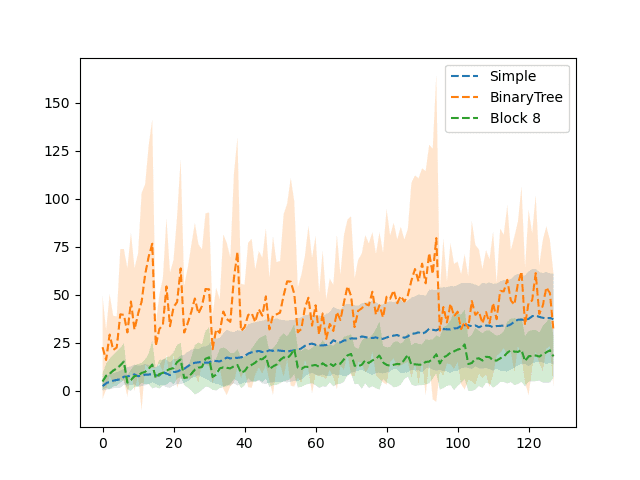}
         \caption{Short time horizon of $2^7 = 128$}
     \end{subfigure}
    \caption{Comparing Simple, Block with size 8, and Binary Tree counters with privacy budget $\e=0.5$. The true data is a random bit stream of 0s and 1s. On the y-axis, we have the absolute error in the value of the counter at any time averaged over 10 independent runs of the counter algorithm.}
    \label{fig:compare_counters}
    \vskip -0.1in
\end{figure}

%% file: appendices/datasets.tex
\section{Datasets and Geometry} \label{app: dataset and geometry}

\subsection{Datasets}

 \begin{figure*}[ht]
    \vskip 0.1in
     \centering
     \begin{subfigure}[t]{\textwidth}
         \centering
         \includegraphics[width=\textwidth]{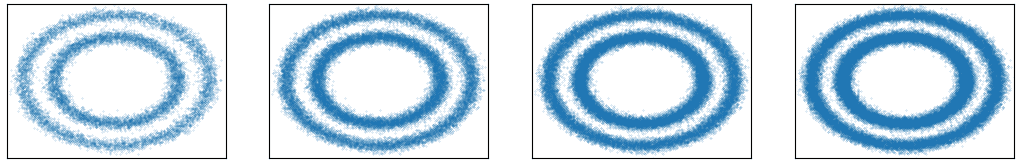}
         \caption{Without deletion, where density on both circles grow with time}
         \label{fig:data_circles_a}
     \end{subfigure}
     \hfill
     \begin{subfigure}[t]{\textwidth}
         \centering
         \includegraphics[width=\textwidth]{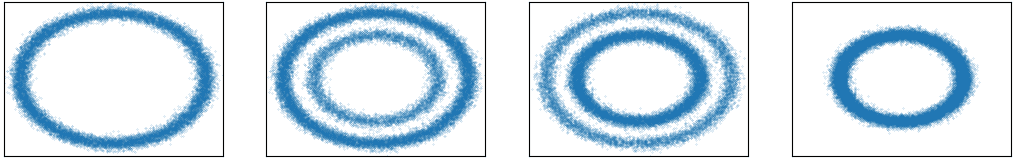}
         \caption{With deletion, where the density gradually changes from predominantly on the outer circle to the inner circle}
         \label{fig:data_circles_b}
     \end{subfigure}
    \caption{Scatter plot for the simulated datasets of two concentric circles showing the resulting location of users at four different times steps progressing from left to right}
    \label{fig:data_circles}
    \vskip -0.1in
\end{figure*}

 \begin{figure*}
    \vskip 0.1in
     \centering
     \begin{subfigure}[t]{0.45\textwidth}
         \centering
         \includegraphics[width=\textwidth]{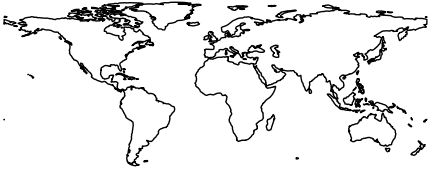}
         \caption{Gowalla}
         \label{fig:geom_gowalla}
     \end{subfigure}
     \hfill
     \begin{subfigure}[t]{0.35\textwidth}
         \centering
         \includegraphics[width=\textwidth]{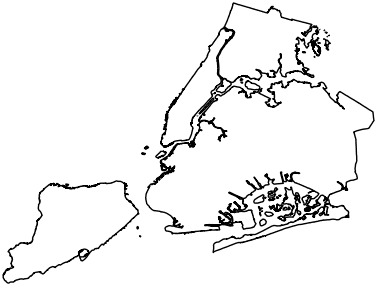}
         \caption{New York taxi (over NY State)}
         \label{fig:geom_ny_state}
     \end{subfigure}
     \hfill
     \begin{subfigure}[t]{0.15\textwidth}
         \centering
         \includegraphics[width=\textwidth]{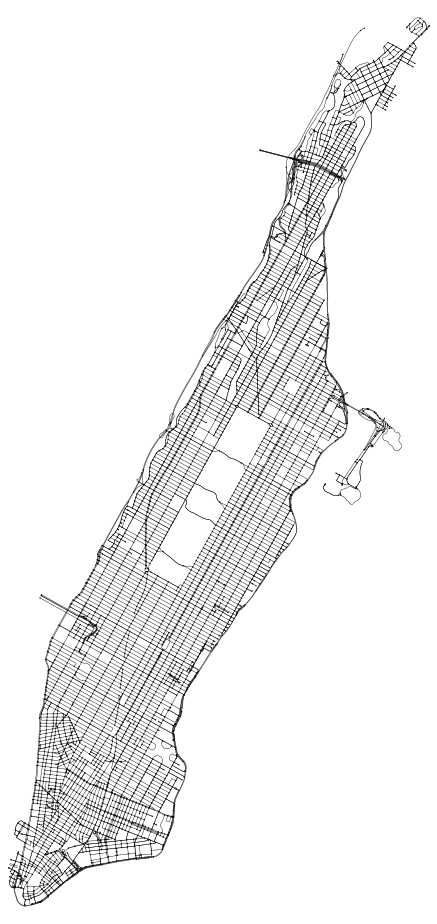}
         \caption{New York Taxi (over road network of Manhattan)}
         \label{fig:geom_ny_mht}
     \end{subfigure}
    \caption{Geometry used for datasets}
    \label{fig:geometries}
    \vskip -0.1in
\end{figure*}

\textbf{Gowalla:} The Gowalla dataset \cite{gowalla} contains location check-ins by a user along with their user id and time on the social media app Gowalla. For this dataset, we use the natural land geometry of Earth except for the continent Antarctica (see Figure~\ref{fig:geom_gowalla} in Appendix~\ref{app: dataset and geometry}). The size of the dataset after restriction to the geometry is about $6.2$ million. We limit the total number of data points to about $200$ thousand and follow a daily release frequency based on the available timestamps.

\textbf{Gowalla with deletion:} To evaluate our method on a dataset with deletion that represents a real-world scenario, we create a version of the Gowalla dataset with deletion. Motivated by the problem of releasing the location of active coronavirus infection cases, we consider each check-in as a report of infection and remove this check-in after $30$ days indicating that the reported individual has recovered.

\textbf{NY taxi:} The NY Taxi dataset \cite{ny_data} contains the pickup and drop-off locations and times of Yellow Taxis in New York. For this dataset, we only use data from January $2013$ which already has more than $14$ million data points. We keep the release frequency to $6$ hours and restrict total data points to about $10^5$ sampling proportional to the number of data points at each time within the overall time horizon. We consider two different underlying geometries for this dataset. The first is the natural land geometry of the NY State as a collection of multiple polygons. The second is the road network of Manhattan Island as a collection of multiple curves. We include a figure of these geometries in Figure~\ref{fig:geom_ny_state} and Figure~\ref{fig:geom_ny_mht} respectively of Appendix~\ref{app: dataset and geometry}. The second geometry is motivated by the fact that the majority of taxi pickup locations are on Manhattan Island and also on its road network.

\textbf{Concentric circles:} We also conducted experiments on an artificial dataset of points located on two concentric circles. As shown in Figure~\ref{fig:data_circles}, we have two different scenarios of this dataset: (a) where both circles grow over time and (b) where the points first appear in one circle and then gradually move to the other circle. Scenario (b) helps us analyze the performance of our algorithm on a dataset where (unlike Gowalla and NY Taxi datasets) the underlying distribution changes dramatically. We also wanted to explore the performance of our algorithm based on the number of data points it observes in initialization and then at each batch. Thus for this dataset, we first fix a constant batch size for the experiment and then index time accordingly. Additionally, we keep initialization time $t_0$ as a value in $[0,1]$ denoting the proportion of total data the algorithm uses for initialization. We explore various parameters on these artificial datasets such as batch size, initialization data ratio, and sensitivity.

\subsection{Initialization Time}

 \begin{figure*}[h]
    \vskip 0.1in
     \centering
     \begin{subfigure}[t]{0.3\textwidth}
        \centering
         \includegraphics[width=\textwidth]{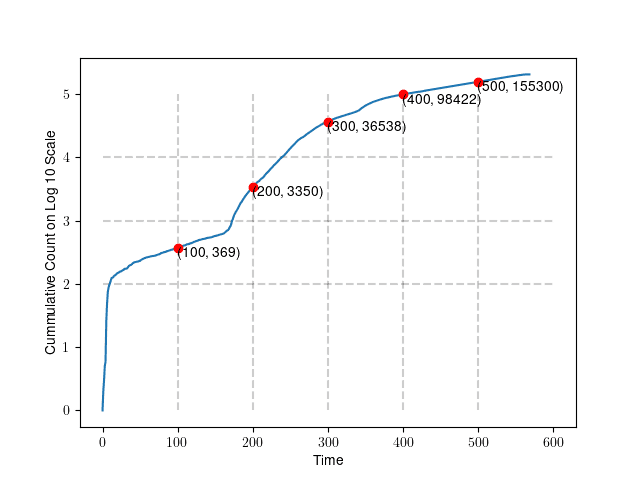}
         \caption{Gowalla}
     \end{subfigure}
     \begin{subfigure}[t]{0.3\textwidth}
        \centering
         \includegraphics[width=\textwidth]{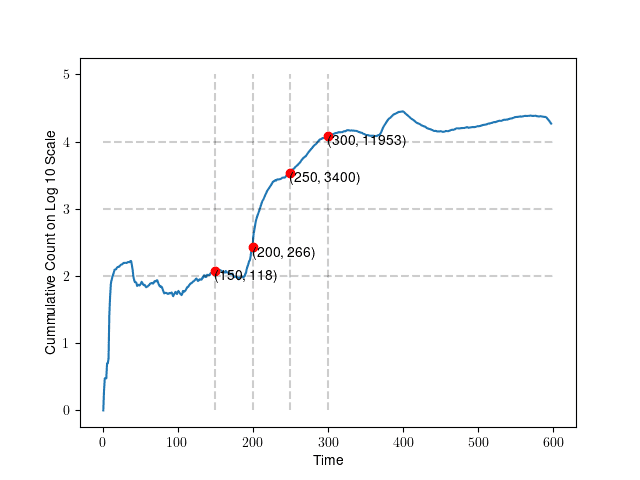}
         \caption{Gowalla with deletion}
     \end{subfigure}
     \begin{subfigure}[t]{0.3\textwidth}  
        \centering
         \includegraphics[width=\textwidth]{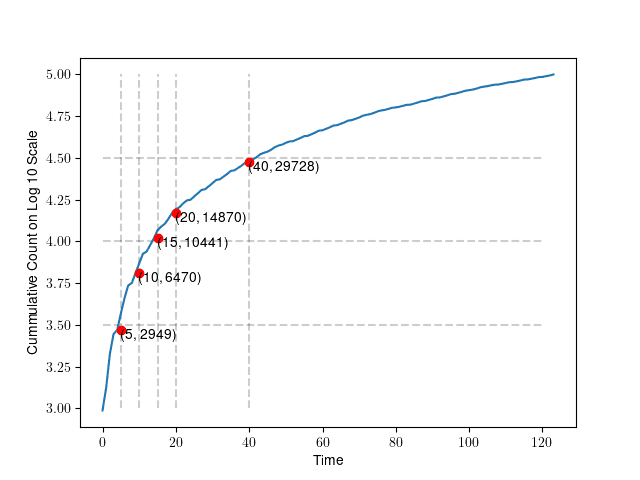}
         \caption{NY Taxi}
     \end{subfigure}
    \caption{Cummulative count as time progresses for Gowalla and NY Taxi Dataset. The plot illustrates our motivation for the choice of initialization time.}
    \label{fig:initialization_time_choice}
    \vskip -0.1in
\end{figure*}

In Figure~\ref{fig:initialization_time_choice}, we show how the total number of data points change with time for the Gowalla and NY Taxi Datasets. For our algorithm to work best, we recommend having points in the order of at least $100$ at each time of the differential stream. Hence, based on the cumulative count observed in Figure~\ref{fig:initialization_time_choice} we select minimum $t_0=100$ for Gowalla, $t_0=150$ for Gowalla with deletion, and $t_0=0$ for NY Taxi datasets.

%% file: appendices/synthetic_data.tex
\section{Synthetic data scatter plots}\label{app: synthetic scatter plots}
In this section, we present a visualization of the synthetic data generated by our algorithm PHDStream with the Simple counter. We present scatter plots comparing the private input data with the generated synthetic data in Figures~\ref{fig:appendix_gowalla_scatter_plots}, \ref{fig:appendix_ny_scatter_plots}, \ref{fig:appendix_ny_mht_scatter_plots}, and \ref{fig:appendix_motion_circle_scatter_plots} for the datasets Gowalla, NY Taxi, and Concentric Circles with deletion respectively. Due to space limitations, we only show the datasets a few times, with time increasing from left to right. In each of the figures mentioned above, the first row corresponds to the private input data and is labeled "True". In each subplot, we plot the coordinates present at that particular time after accounting for all additions and deletions so far. Each of the following rows corresponds to one run of the PHDStream algorithm with the Simple counter and for a particular privacy budget $\e$ as labeled on the row. In each subplot, we use black dashed lines to create a partition of the domain corresponding to the leaves of the current subtree as generated by $\PrivTree_T$. For comparison, we overlay the partition created by the algorithm with $\e=0.5$ on the true data plot as well.

 \begin{figure}[h]
     \centering
     \includegraphics[width=\textwidth,keepaspectratio]{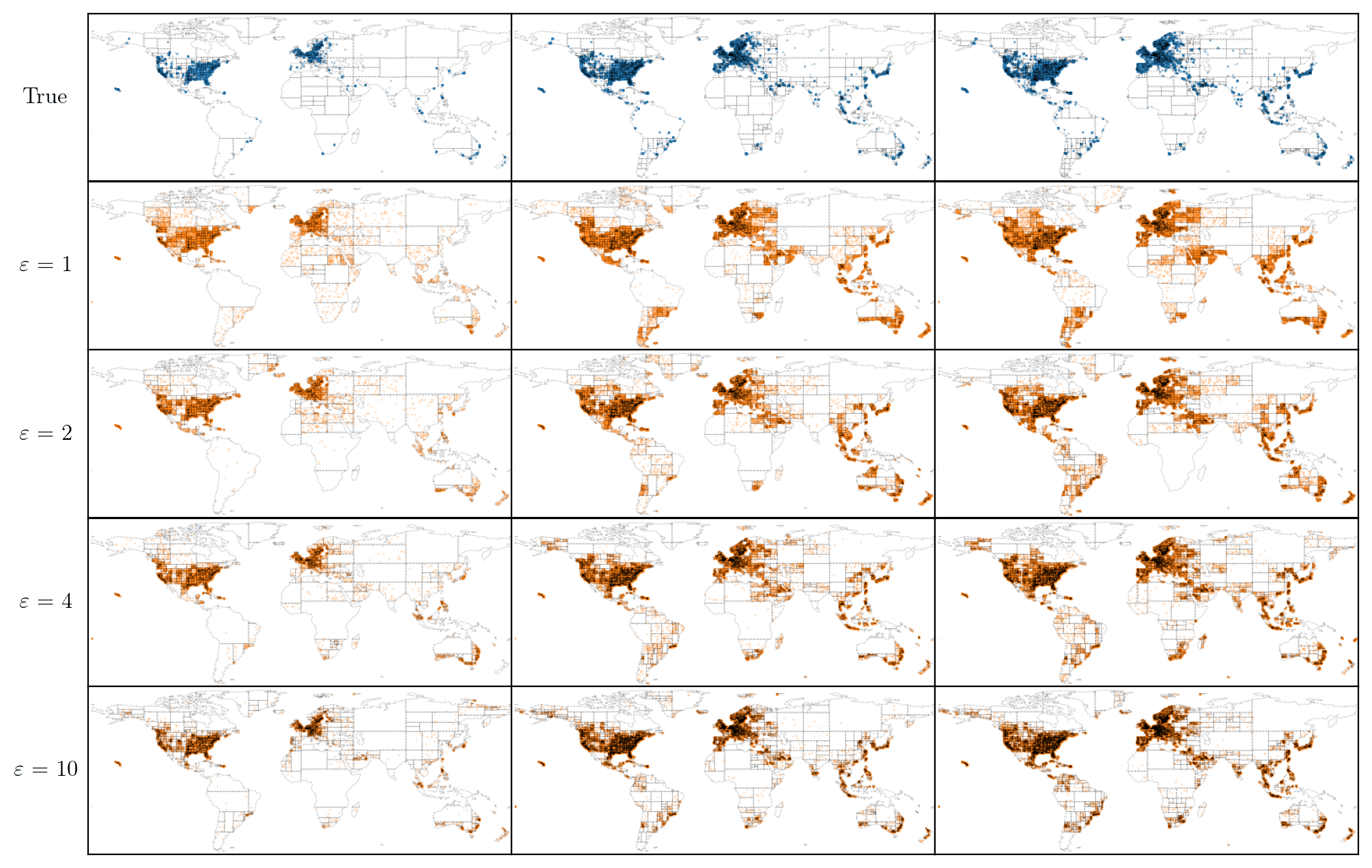}
    \caption{Scatter plot of True and Synthetic data (generated by PHDStream) over the Gowalla dataset for initialization time index $t_0=100$}
    \label{fig:appendix_gowalla_scatter_plots}
\end{figure}
 \begin{figure}[b]
     \centering
     \includegraphics[width=\textwidth]{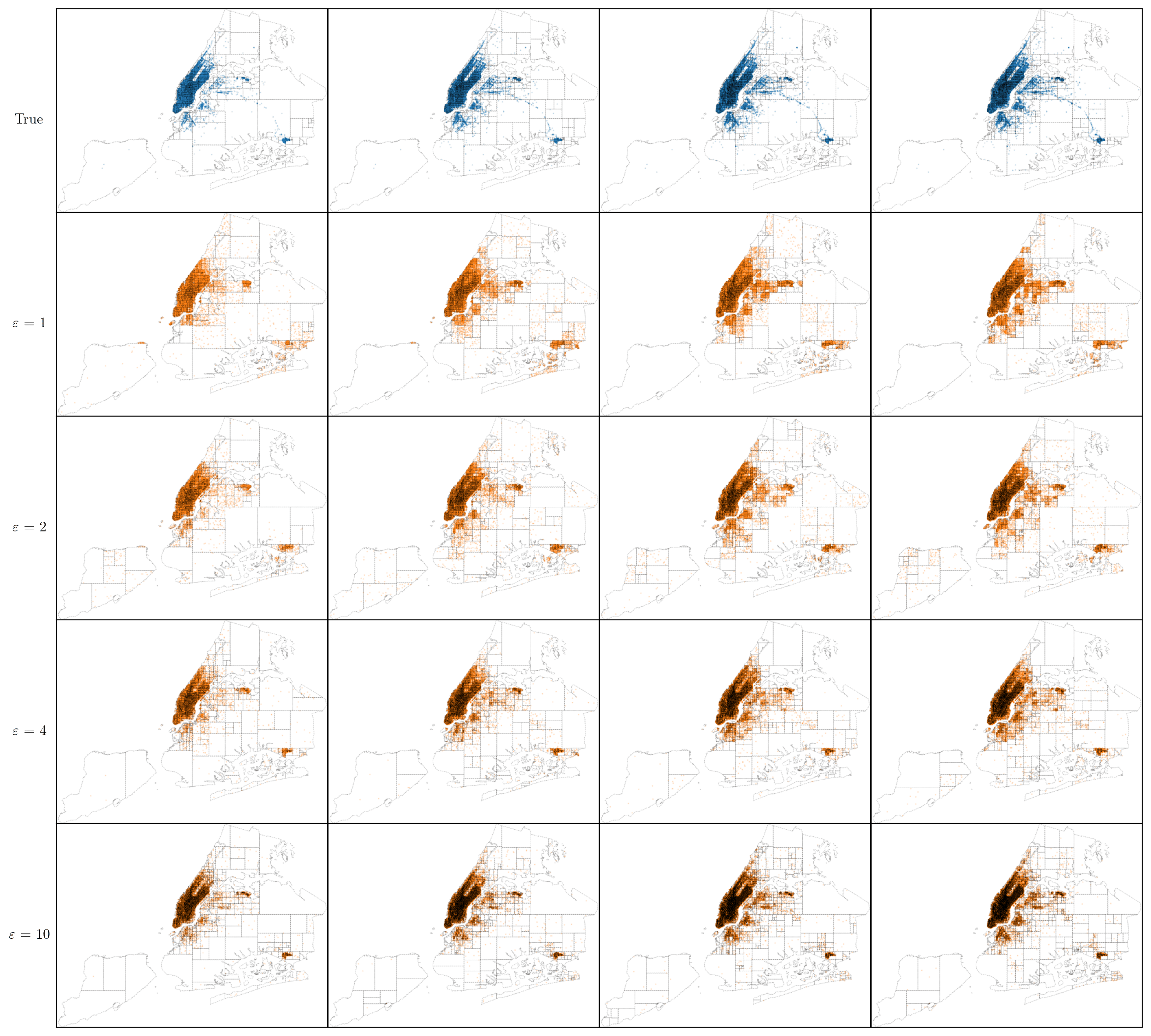}
    \caption{Scatter plot of True and Synthetic data (generated by PHDStream) over the New York dataset over NY State for initialization time index $t_0=2$}
    \label{fig:appendix_ny_scatter_plots}
\end{figure}
 \begin{figure}[b]
     \centering
     \includegraphics[width=\textwidth]{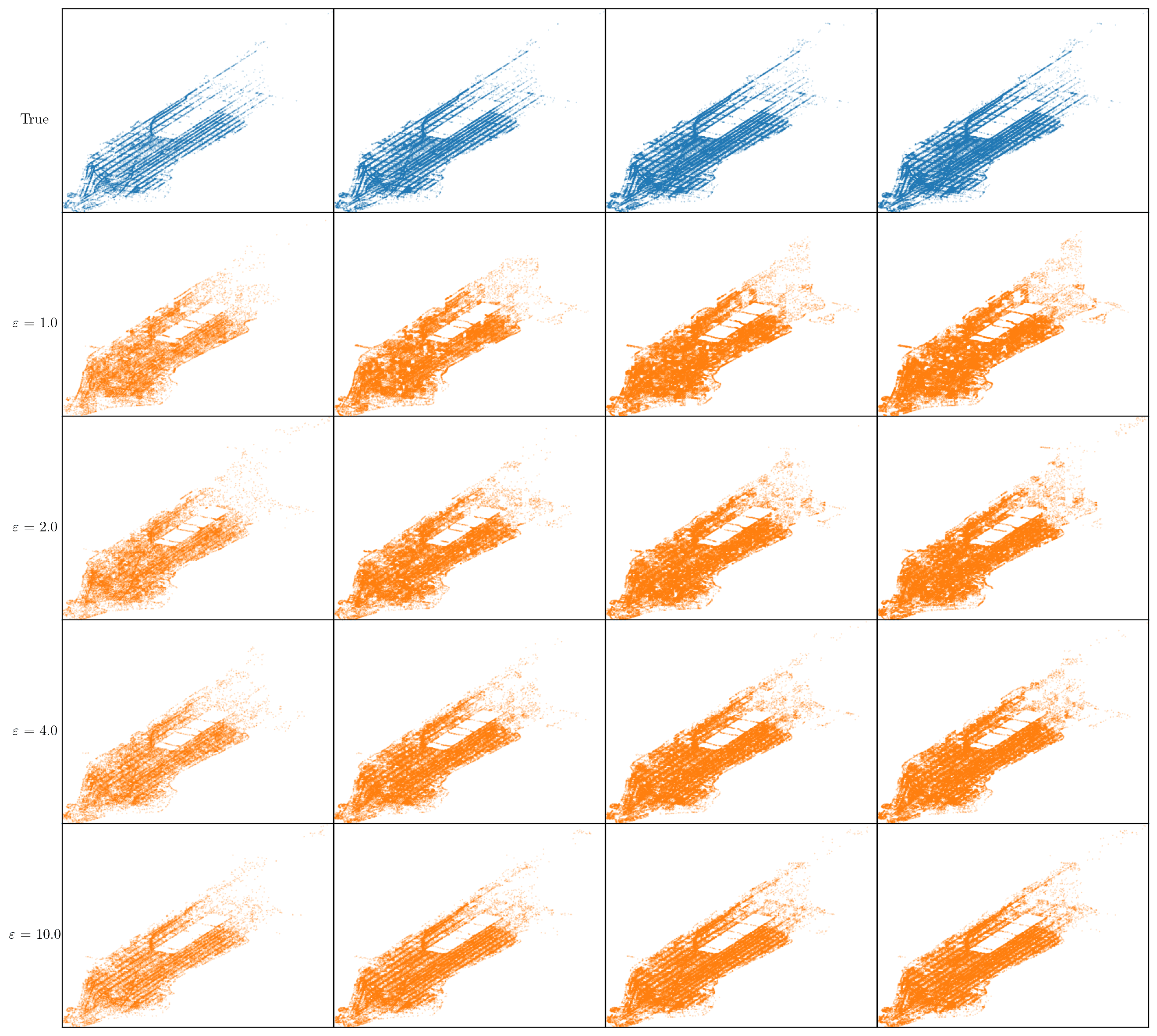}
    \caption{Scatter plot of True and Synthetic data (generated by PHDStream) over the New York dataset over road network of Manhattan for initialization time index $t_0=2$}
    \label{fig:appendix_ny_mht_scatter_plots}
\end{figure}
\begin{figure}
     \centering
    \includegraphics[width=\textwidth,keepaspectratio]{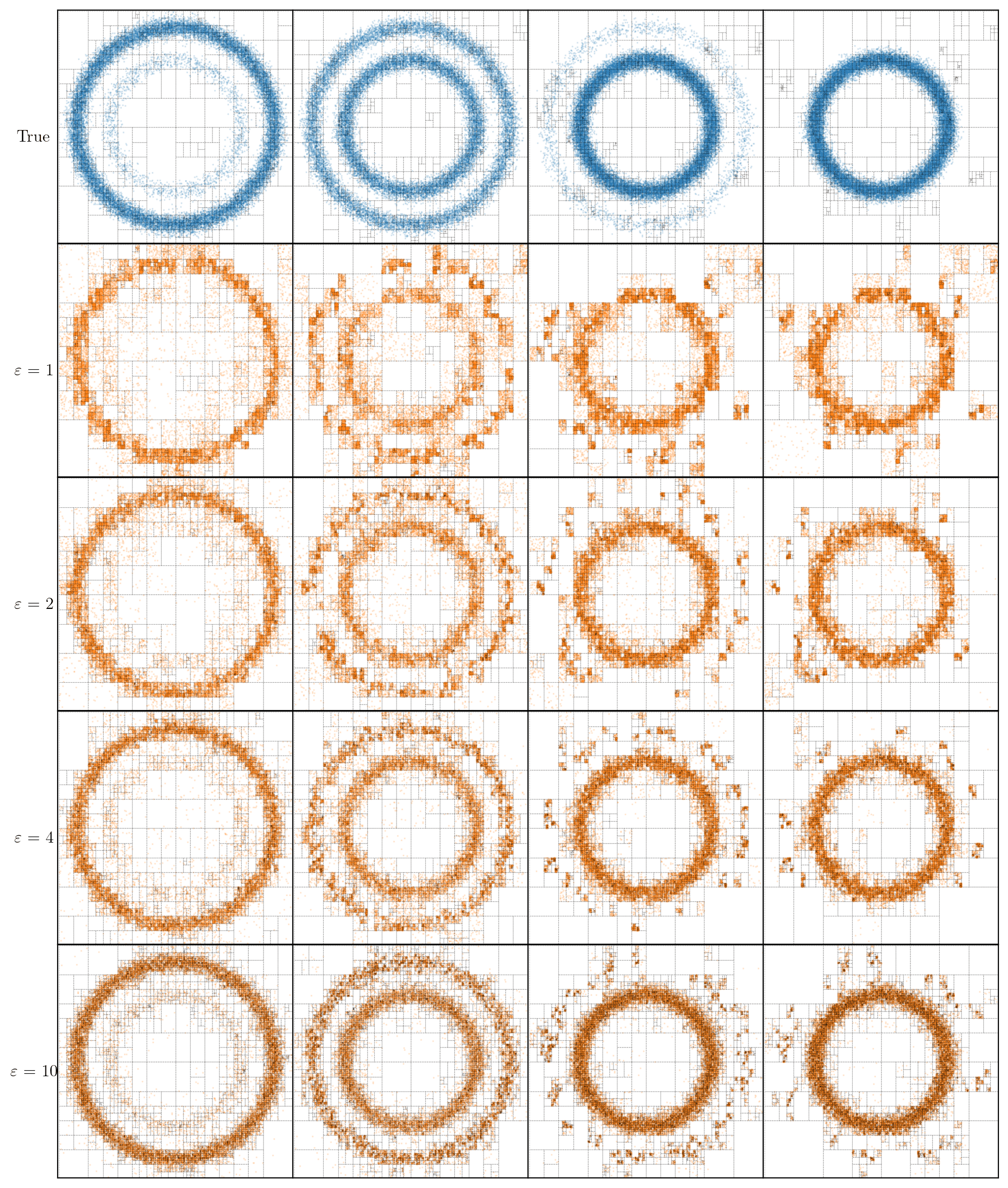}
    \caption{Scatter plot of True and Synthetic data (generated by PHDStream) over the Concentric circles with deletion dataset for initialization data ratio $t_0=0.1$, constant batch size $500$}
    \label{fig:appendix_motion_circle_scatter_plots}
\end{figure}

%% file: appendices/more_results.tex
\section{More Results}\label{app: more results}

Figures~\ref{fig:appendix_gowalla_small}, \ref{fig:appendix_gowalla_with_deletion_small}, \ref{fig:appendix_ny_small}, \ref{fig:appendix_ny_mht_small}, \ref{fig:appendix_no_deletion_circle_small}, and \ref{fig:appendix_motion_circle_small} show a comparison of performances on different datasets for small-range queries. Each subplot has a time horizon on the x-axis and corresponds to a particular value of privacy budget $\e$ (increasing from left to right) and initialization time $t_0$ (increasing from top to bottom). We do not show PHDStream with Block counter and Baseline~2 in these figures as they both have very large errors in some cases and that distorts the scale of the figures. We also show some results on medium and large scale range queries (as described in Section~\ref{s: metric}) for the dataset Gowalla in Figures~\ref{fig:appendix_gowalla_medium}, ~\ref{fig:appendix_gowalla_large} and for the dataset Gowalla with deletion in Figures~\ref{fig:appendix_gowalla_with_del_medium}, ~\ref{fig:appendix_gowalla_with_del_large} respectively.

 \begin{figure}
     \centering         \includegraphics[height=0.45\textheight,keepaspectratio]{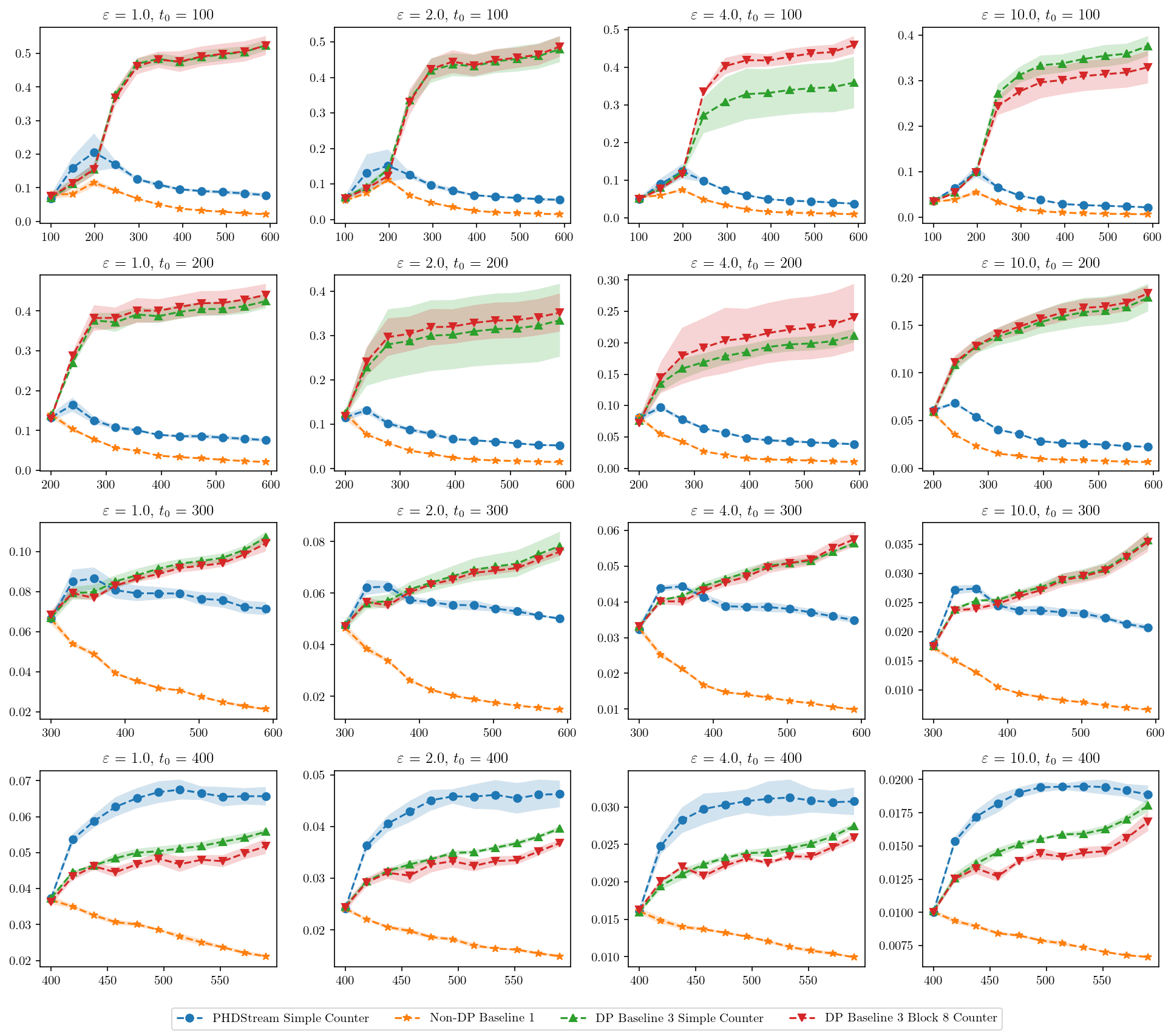}
    \caption{Query error for PHDStream over the dataset Gowalla and small range queries.}
    \label{fig:appendix_gowalla_small}
\end{figure}
 \begin{figure}
     \centering         \includegraphics[height=0.45\textheight,keepaspectratio]{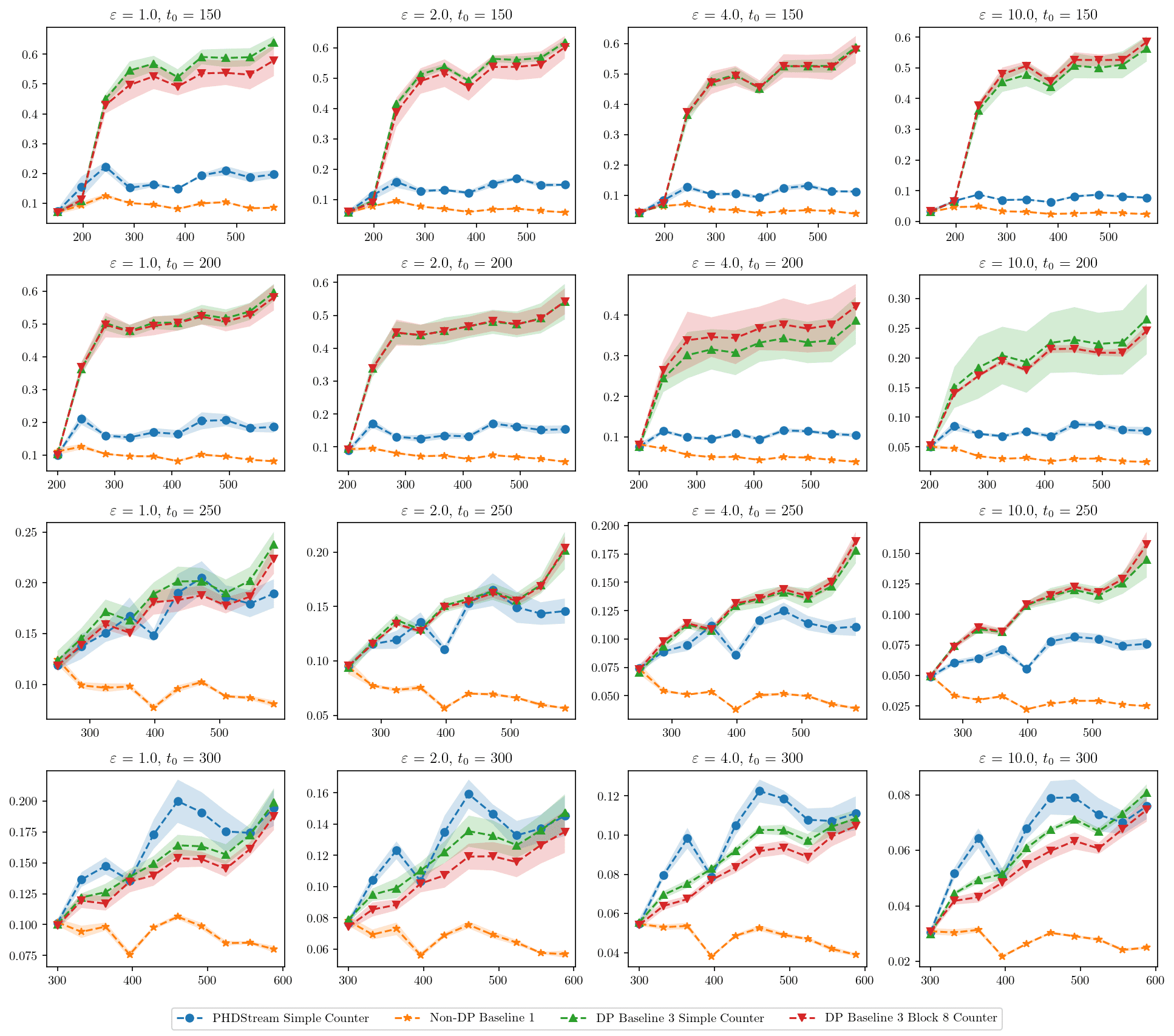}
    \caption{Query error for PHDStream over the dataset Gowalla with deletion and small range queries.}
    \label{fig:appendix_gowalla_with_deletion_small}
\end{figure}

 \begin{figure}
     \centering         \includegraphics[height=0.45\textheight,keepaspectratio]{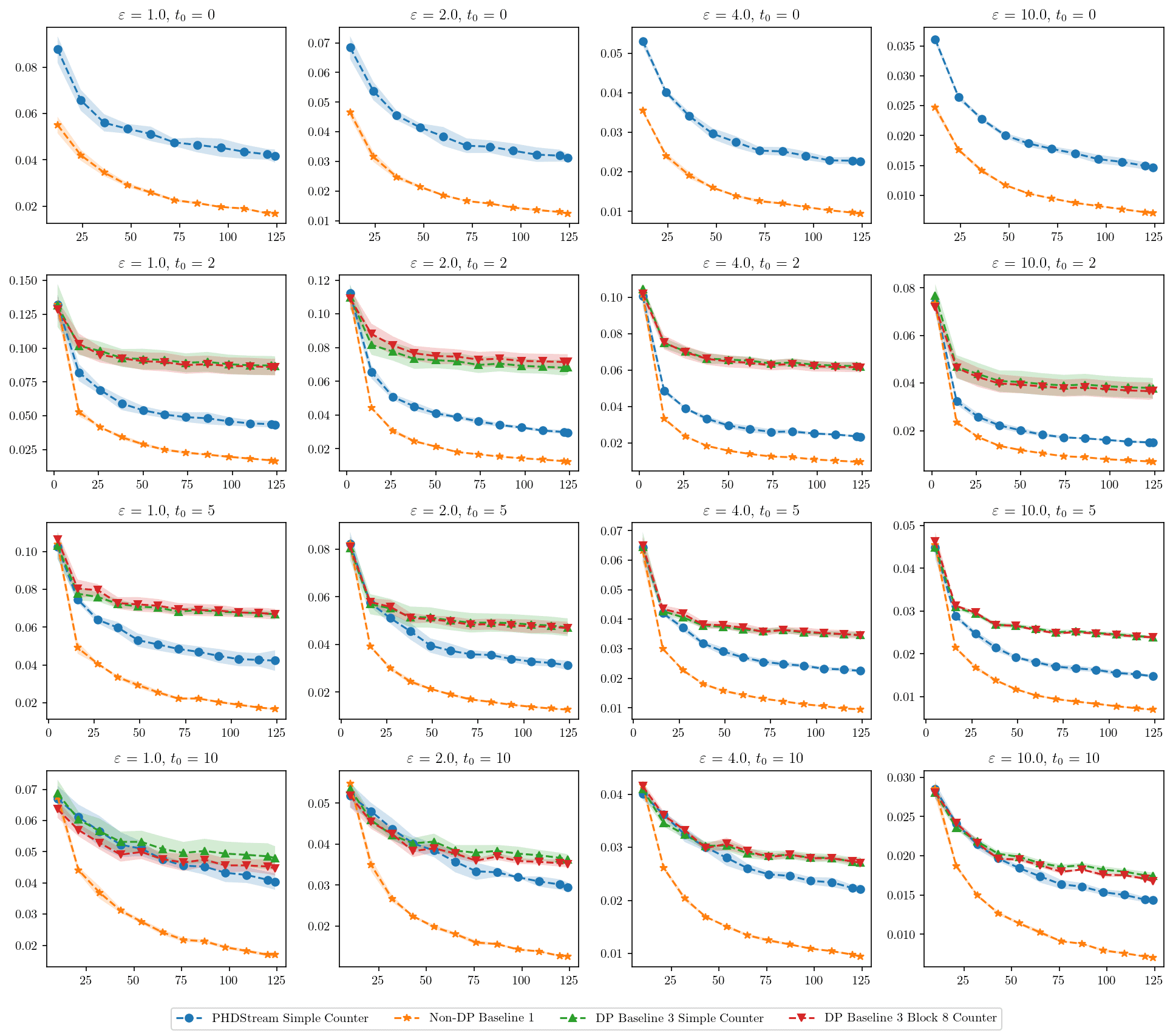}
    \caption{Query error for PHDStream over the dataset New York (over NY state) and small range queries.}
    \label{fig:appendix_ny_small}
\end{figure}

 \begin{figure}
     \centering         \includegraphics[height=0.45\textheight,keepaspectratio]{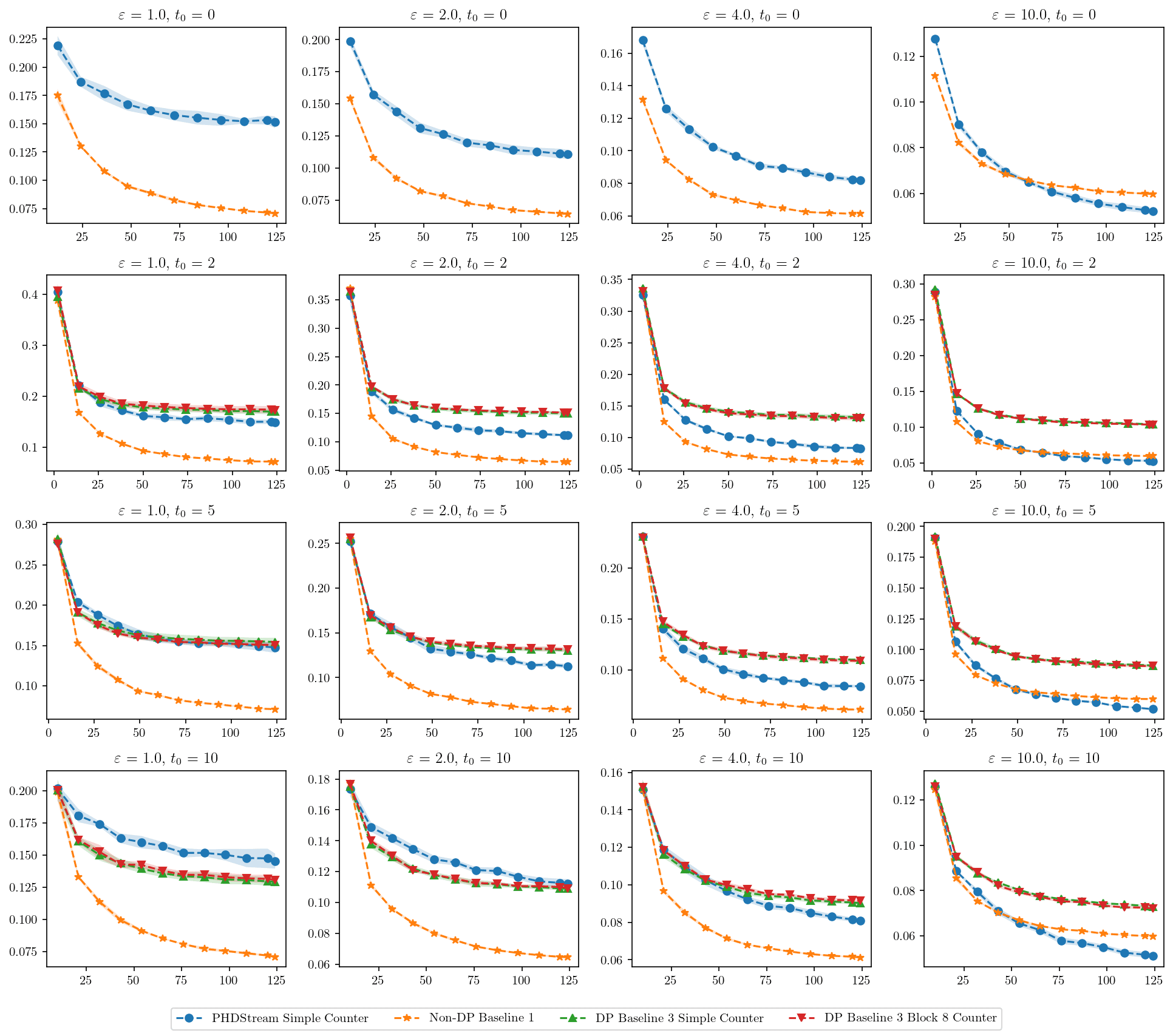}
    \caption{Query error for PHDStream over the dataset New York (over Manhattan road network) and small range queries.}
    \label{fig:appendix_ny_mht_small}
\end{figure}

\begin{figure}
    \centering
    \includegraphics[height=0.45\textheight,keepaspectratio]{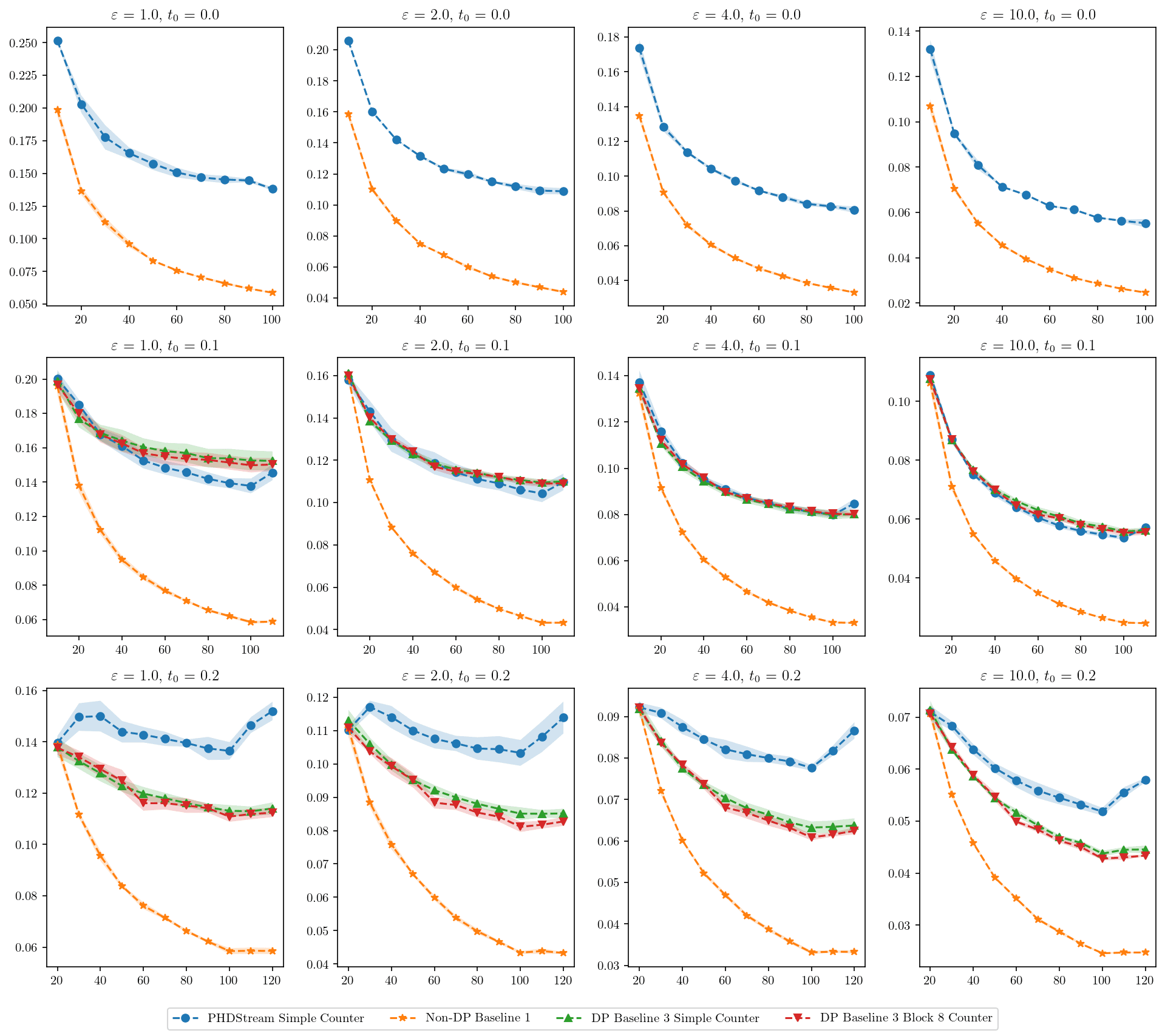}
    \caption{Query error for PHDStream over the dataset Concentric circles no deletion and small range queries.}
    \label{fig:appendix_no_deletion_circle_small}
\end{figure}
\begin{figure}
     \centering
     \includegraphics[height=0.45\textheight,keepaspectratio]{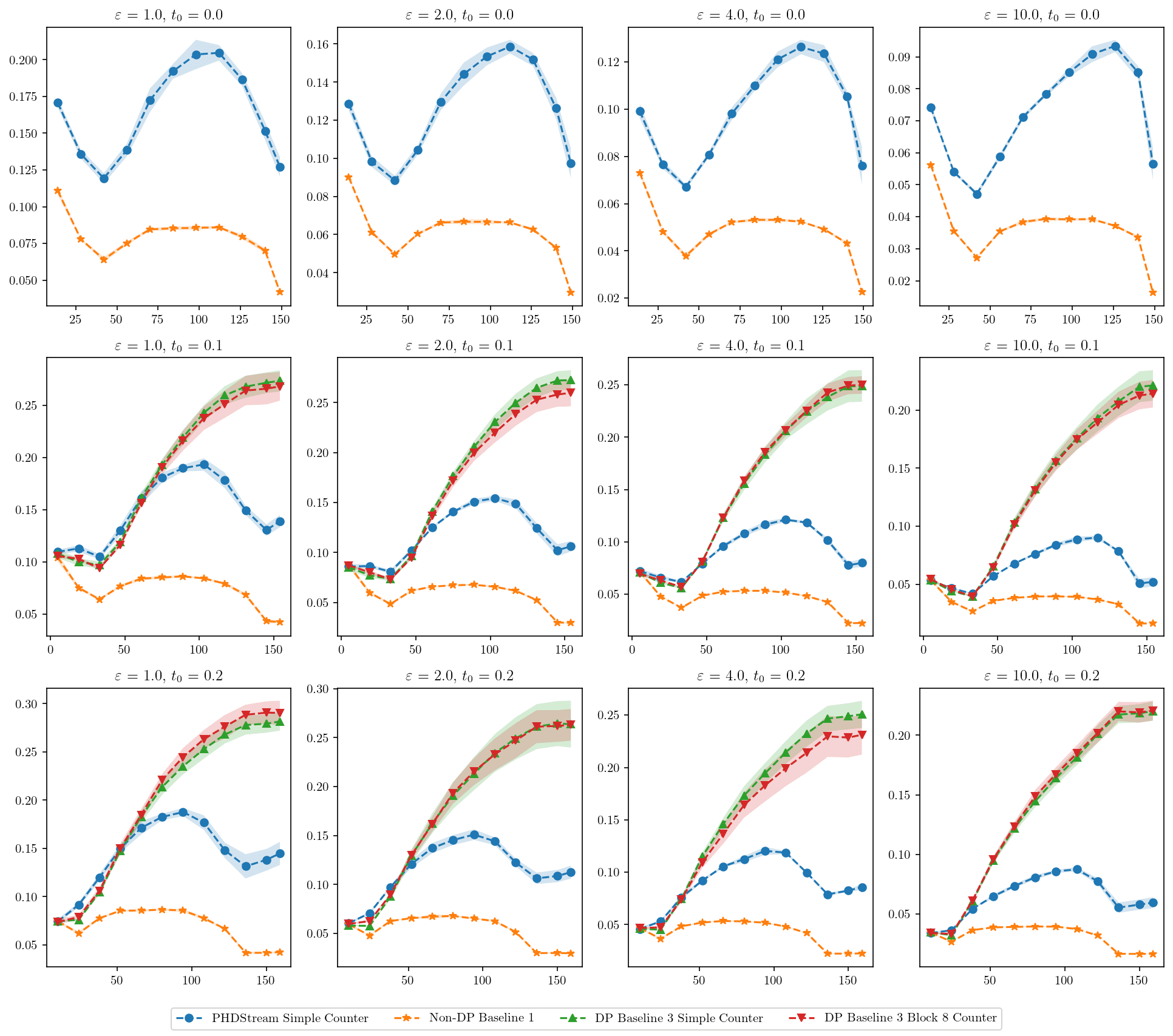}
    \caption{Query error for PHDStream over the dataset Concentric circles with deletion and small range queries.}
    \label{fig:appendix_motion_circle_small}
\end{figure}

 \begin{figure}
     \centering         \includegraphics[height=0.45\textheight,keepaspectratio]{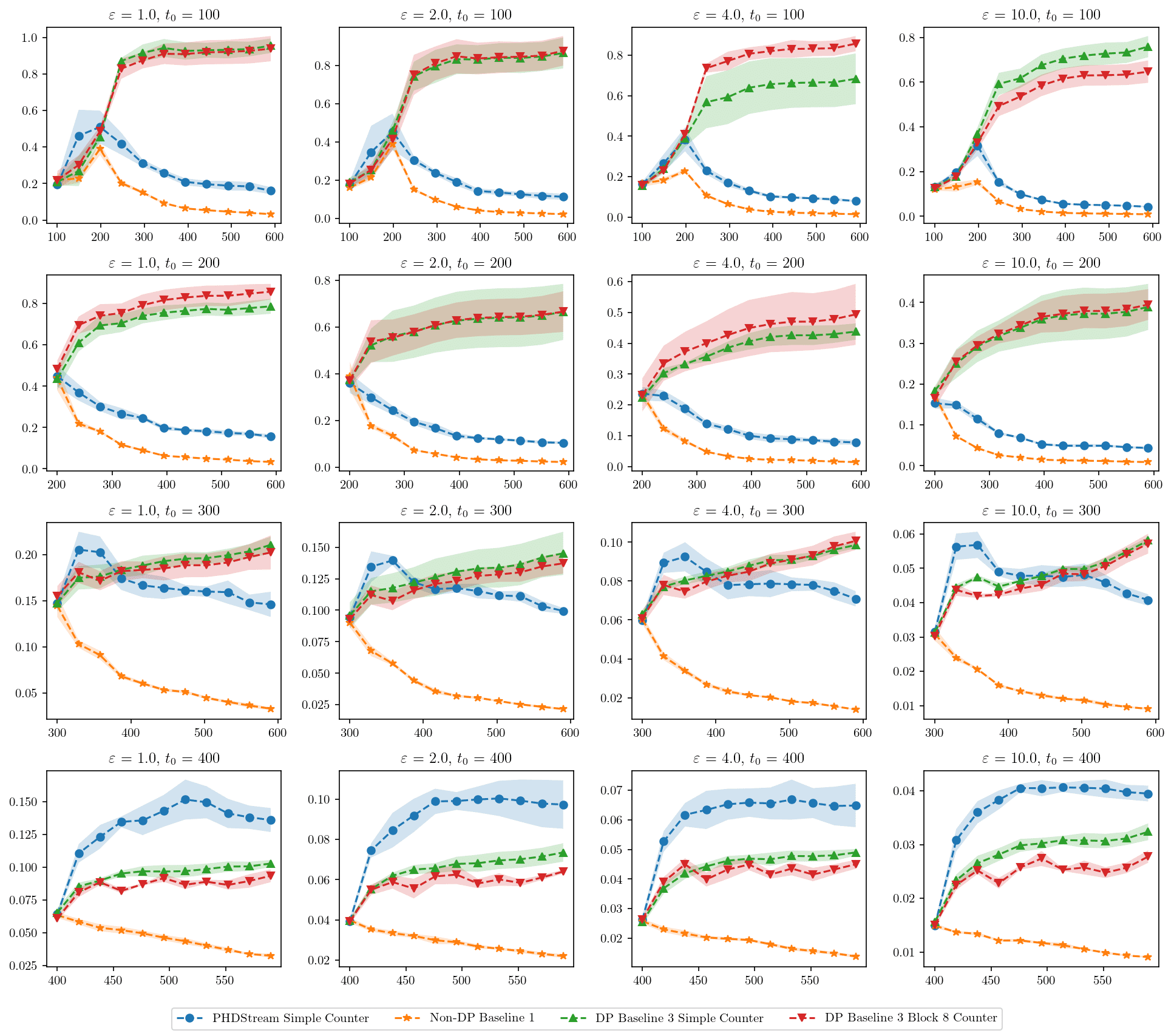}
    \caption{Query error for PHDStream over the dataset Gowalla and medium range queries.}
    \label{fig:appendix_gowalla_medium}
\end{figure}
 \begin{figure}
     \centering         \includegraphics[height=0.45\textheight,keepaspectratio]{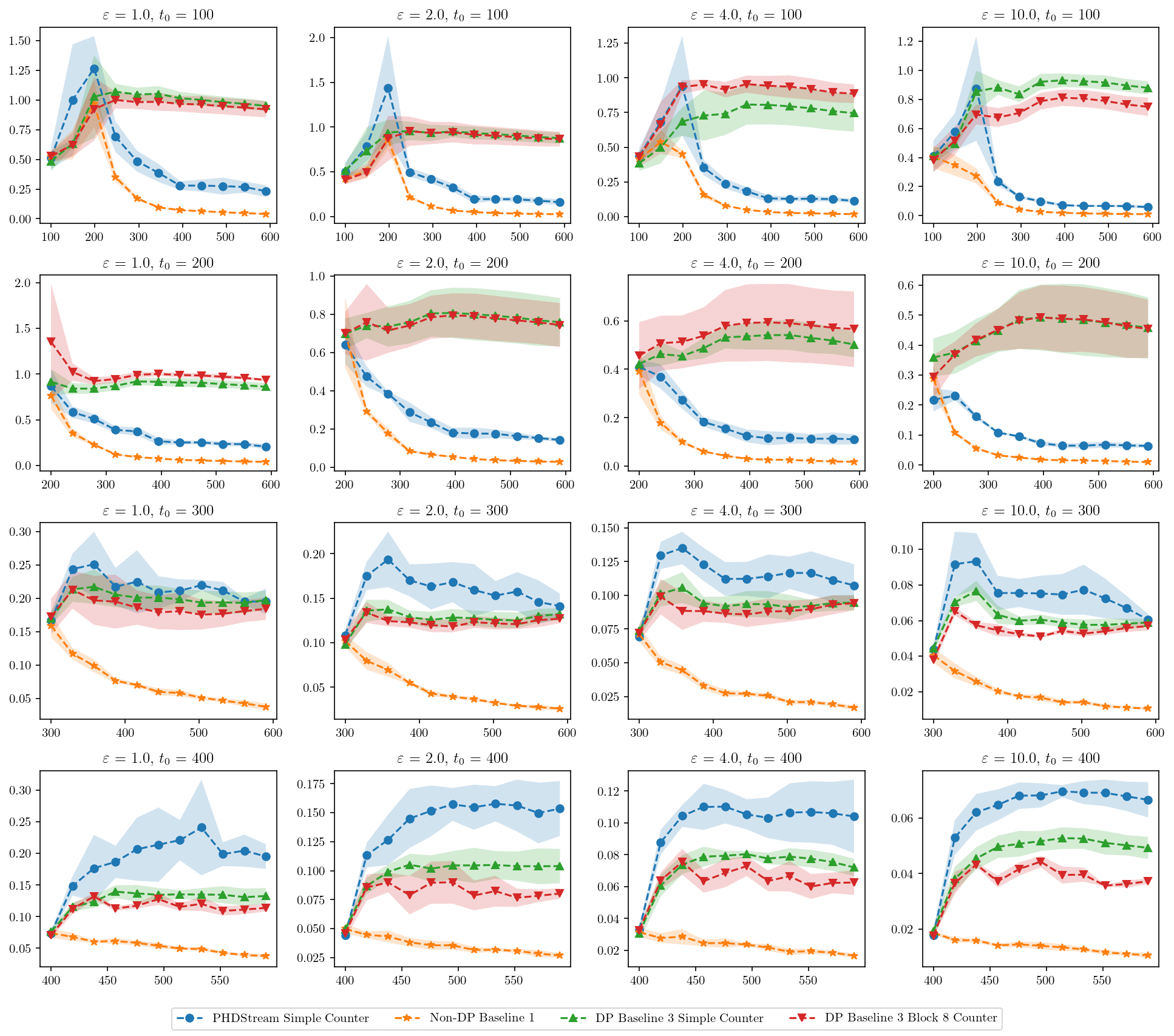}
    \caption{Query error for PHDStream over the dataset Gowalla and large range queries.}
    \label{fig:appendix_gowalla_large}
\end{figure}

 \begin{figure}
     \centering         \includegraphics[height=0.45\textheight,keepaspectratio]{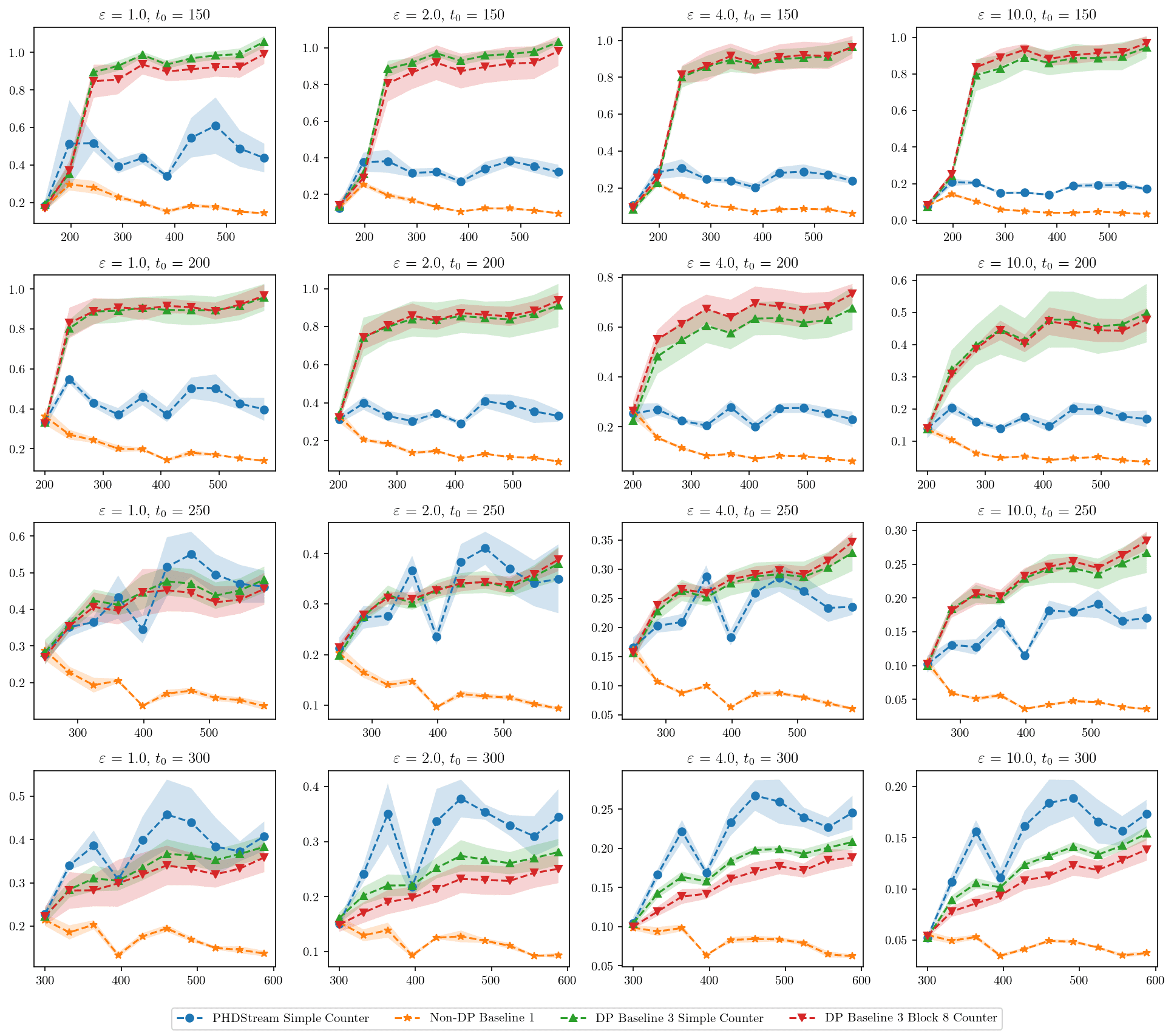}
    \caption{Query error for PHDStream over the dataset Gowalla with deletion and medium range queries.}
    \label{fig:appendix_gowalla_with_del_medium}
\end{figure}
 \begin{figure}
     \centering         \includegraphics[height=0.45\textheight,keepaspectratio]{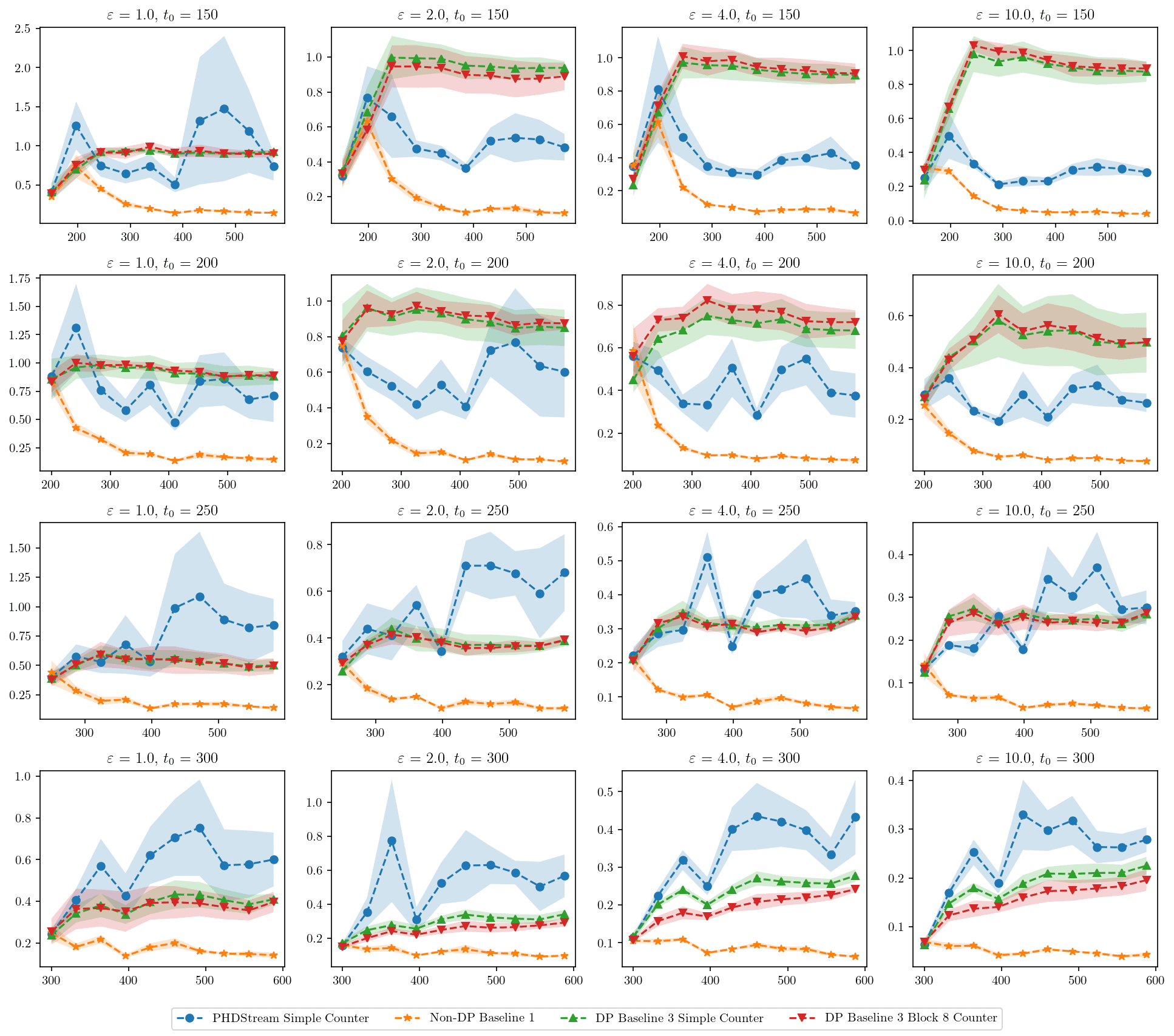}
    \caption{Query error for PHDStream over the dataset Gowalla with deletion and large range queries.}
    \label{fig:appendix_gowalla_with_del_large}
\end{figure}